\documentclass[10pt,twocolumn,twoside]{IEEEtran}
\usepackage{color}
\usepackage{algorithm}
\usepackage{algpseudocode}
\usepackage{undertilde}

\usepackage{ifpdf}

\usepackage{cite}

\ifCLASSINFOpdf
   \usepackage[pdftex]{graphicx}

\else
  \usepackage[dvips]{graphicx}
\fi
\usepackage[cmex10]{amsmath}
\usepackage{array}
\usepackage{eqparbox}

\usepackage{stfloats}
\usepackage{url}

\usepackage{color}
\usepackage{multirow}
\usepackage{amsmath,amsfonts}
\usepackage[colorlinks]{hyperref}
\usepackage{subfigure}
\usepackage{caption2}
\usepackage{graphicx}
\usepackage{tikz}

\newcommand{\pgftextcircled}[1]{
    \setbox0=\hbox{#1}%
    \dimen0\wd0%
    \divide\dimen0 by 2%
    \begin{tikzpicture}[baseline=(a.base)]%
        \useasboundingbox (-\the\dimen0,0pt) rectangle (\the\dimen0,1pt);
        \node[circle,draw,outer sep=0pt,inner sep=0.1ex] (a) {#1};
    \end{tikzpicture}
}

\newcommand{\pgftextcircledblk}[1]{
    \setbox0=\hbox{#1}%
    \dimen0\wd0%
    \divide\dimen0 by 2%
    \begin{tikzpicture}[baseline=(a.base)]%
        \useasboundingbox (-\the\dimen0,0pt) rectangle (\the\dimen0,1pt);
        \node[circle,draw,outer sep=0pt,inner sep=0.1ex,fill=blue] (a) {#1};
    \end{tikzpicture}
}

\newcommand{\red}[1]{
       \textcolor{red}{#1}
}

\def\supp{{\rm supp}}

\def\t0{{t_0}}

\def\R{{\mathbb R}}        

\def\t0{{t_0}}

\def\vs{{\mathbf s}}

\def\vE{{\mathbf E}}
\def\vL{{\mathbf L}}
\def\L{{\mathcal L}}
\def\vb{{\mathbf b}}
\def\vD{{\mathbf D}}        
\def\vA{{\mathbf A}}        
\def\R{{\mathbb R}}        

\DeclareMathOperator*{\Min}{minimize}
  \newtheorem{theorem}{Theorem}
    \newtheorem{lemma}{Lemma}
      
  \newtheorem{remark}{Remark}

\begin{document}
%
\title{Fast Adaptive Algorithm for Robust Evaluation of Quality of Experience}
%
%
%

\author{Qianqian~Xu,
        ~Ming~Yan,
        and~Yuan~Yao

\thanks{Q. Xu is with BICMR, Peking University, Beijing 100871, China, (email: xuqianqian@math.pku.edu.cn).}
\thanks{M. Yan is with Department of Mathematics, University of California, Los Angeles, CA 90095, USA, (email: yanm@math.ucla.edu).}
\thanks{Y. Yao is with LMAM-LMP-LMEQF, School of Mathematical Sciences, Peking University, Beijing 100871, China, (email: yuany@math.pku.edu.cn).}

\thanks{This work is supported in part by National Basic Research Program of China (973 Program 2012CB825501), NSFC Grant 61071157 and 61402019, NSF Grants DMS-1349855, DMS-1317602, and ARO MURI Grant W911NF-09-1-0383.
}
}

%
%

\markboth{}%
{XU \MakeLowercase{\textit{et al.}}: iLTS for Outlier Detection}
%



\maketitle

\begin{abstract}
Outlier detection is an integral part of robust evaluation for crowdsourceable Quality of Experience (QoE) and has attracted much attention in recent years. In QoE for multimedia, outliers happen because of different test conditions, human errors, abnormal variations in context, {etc}. In this paper, we propose a simple yet effective algorithm for outlier detection and robust QoE evaluation named iterative Least Trimmed Squares (iLTS). The algorithm assigns binary weights to samples, i.e., 0 or 1 indicating if a sample is an outlier, then the outlier-trimmed subset least squares solutions give robust ranking scores. An iterative optimization is carried alternatively between updating weights and ranking scores which converges to a local optimizer in finite steps. In our test setting, iLTS is up to 190 times faster than LASSO-based methods with a comparable performance. Moreover, a varied version of this method shows adaptation in outlier detection, which provides an automatic detection to determine whether a data sample is an outlier without \emph{a priori} knowledge about the amount of the outliers. The effectiveness and efficiency of iLTS are demonstrated on both simulated examples and real-world applications. A Matlab package is provided to researchers exploiting crowdsourcing paired comparison data for robust ranking.

%
%
%

\end{abstract}


\section{INTRODUCTION}

In recent years, the quality of experience (QoE) notion~\cite{Hossfeld12-QoE,Wu13crowd} has become a major research theme within the multimedia community, which can be described as the assessment of a user's subjective expectation, feeling, perception, and satisfaction with respect to multimedia content. There are two main quality assessment methodologies, namely subjective and objective assessment. Measuring and ensuring good QoE of multimedia content is highly subjective in nature.
The most commonly used subjective method for quality measurement is the mean opinion score
(MOS). MOS is standardized in the ITU-T recommendations~\cite{MOS}, and it is defined as a numeric
value going from 1 to 5 (i.e., bad to excellent). Although the MOS rating method has a long history of pervasive use, it suffers from three fundamental problems: (i) Unable to concretely define the concept of scale; (ii) Dissimilar interpretations of the scale among users; (iii) Difficult to verify whether a participant gives false ratings
either intentionally or carelessly~\cite{MM09}.

Therefore, to address the problems above, we turn to an alternative approach by leveraging the pairwise preference information (i.e., pairwise comparison) obtained from raters. Pairwise comparison has a long history, dating
back to the $18^{th}$ century. It also has many nice properties. For example, pairwise comparison is
a relative measure which is easier to conduct than absolute rating scores and it helps reduce bias from the rating scale. In Netflix dataset, the rating matrix is 99\% incomplete, whereas the paired comparison matrix is only 0.22\%
incomplete and most entries are supported by many comparisons~\cite{Hodge}. In some cases such as tennis tournaments, even only pairwise comparison is possible. However, since the number of pairs $n \choose 2$ grows quadratically with the number of alternatives under investigation, this approach may be an expensive and time-consuming process in a laboratory setting.

To meet this challenge, with the advent of ubiquitous Internet access, the crowdsourcing strategy arises to be a promising alternative approach~\cite{Crowdsourcing}. It provides an easy and relatively inexpensive way to accomplish small and simple tasks, such
as Human Intelligence Tasks (HITs), and to effectively utilize the wisdom of the commons to solve complicated projects. Typically,
in a crowdsourcing scenario, each individual contributor is asked to solve a part of a big problem, and a computational algorithm is then developed to combine the partial solutions into an integrated one. Because of the considerable size of the Internet crowd, crowdsourcing could provide us efficient and reliable QoE assessments taking advantage of the power of the mass~\cite{Wu13crowd}.

Methods for rating/ranking via pairwise comparison in QoE evaluation in crowdsourcing scenario must address a number of inherent difficulties including: (i) incomplete and imbalanced data; (ii) streaming and online data; (iii) outlier detection. To meet the first challenge, the work in~\cite{added,MM11,tmm12} propose randomized paired comparison methods which accommodate incomplete and imbalanced data, a general framework called \emph{HodgeRank on random graphs} (HRRG). It not only can deal with incomplete and imbalanced data collected from crowdsourcing studies but also derives the constraints on sampling complexity in crowdsourcing experiment that the random selection must adhere to. Furthermore, a recent extension of HRRG is introduced in~\cite{MM12,TMM13} to deal with streaming and online data in crowdsourcing scenario in the second challenge, providing the possibility of making assessment procedure significantly faster than~\cite {MM11, tmm12} without deteriorating the accuracy.

The third challenge of crowdsourcing QoE evaluations is the fact that not every Internet user is trustworthy. In other words, due to the lack of supervision when subjects perform experiments in crowdsourcing, they may provide erroneous responses perfunctorily, carelessly, or dishonestly~\cite{MM09}. Such random decisions are useless and may deviate significantly from other raters' decisions. Such outliers have to be identified to achieve a robust QoE evaluation. In~\cite{MM09}, Transitivity Satisfaction Rate (TSR), which checks all the intransitive triangles, e.g., $A \succ B \succ C\succ A$, is proposed for outlier detection. TSR is defined as the number of judgment triplets (e.g., the three preference relations among A, B, and C) satisfying transitivity divided by the total number of triplets where transitivity may apply; thus, the value of TSR is always between 0 and 1. If a participant's judgments are consistent throughout all the rounds of an experiment, TSR will be 1; otherwise it will be smaller than 1. In this way, we can identify and discard noisy data provided by unreliable assessors. However, TSR can only be applied for complete and balanced paired comparison data. When the paired data are incomplete and imblanced, {i.e.}, having missing edges, the question of how to detect the noisy pairs remains open. The work in~\cite{MM13} attacks this problem and formulates the outlier detection as a LASSO problem based on sparse approximations of cyclic ranking projection of paired comparison data in Hodge decomposition. Regularization paths of the LASSO problem could provide an order on samples tending to be outliers. However, as every sample contributes an outlier indicator variable, solving such a large scale LASSO is expensive, not mentioning the additional cost on model selection via cross-validation, AIC (Akaike Information Criterion), or BIC (Bayesian Information Criterion) which may not even work well in outlier detection~\cite{SheOwe11}.

In this paper, we propose a simple yet effective algorithm for outlier detection and robust ranking via iterative Least Trimmed Squares (iLTS). This new method is fast, about 190 times faster than LASSO in our test, and adaptive, which could purify data automatically without \emph{a priori} knowledge on the amount of outliers. In our experimental studies on both simulated and real-world data, the method provide comparable results to LASSO in both outlier detection and robust evaluation scores. Therefore it is a promising tool for crowdsourcing robust QoE evaluation.

The remainder of this paper is organized as follows. Section~\ref{sec:relatedwork}
contains a review of related work. Then we describe the proposed framework in Section~\ref{sec:Outlier_Detection}, which establishes some fast and adaptive algorithms based on iterative least trimmed squares. Detailed experiments are presented in Section
\ref{sec:experiments}, followed by the conclusions in Section~\ref{sec:conclusions}.

\section{RELATED WORK}\label{sec:relatedwork}
\subsection{QoE Evaluation}

QoE of multimedia content can be divided into two categories: subjective assessment and objective assessment. In subjective viewing tests, stimuli are shown to a group of viewers, and then their opinions are recorded and averaged to evaluate the quality of the stimuli. This process is labor-intensive and time-consuming. On the contrary, objective assessment predicts the perceived quality automatically and intelligently by building objective quality models (see~\cite{lin}, a survey paper, and its references). Objective methods are indeed convenient to use, whereas it can not capture the \emph{true feelings} of users' experiences. Therefore, to obtain factual QoE evaluation results, subjective methods are still required, even though the cost is higher.

A variety of approaches can be employed to conducting subjective tests, among which mean opinion score (MOS)~\cite{MOS} and paired comparison are the two most popular ones. In the MOS test, individuals are asked to specify a rating from ``Bad" to ``Excellent" (e.g., Bad-1, Poor-2, Fair-3, Good-4, and Excellent-5) to grade the quality of a stimulus; while in paired comparison approach, raters are only asked to make intuitive comparative judgements instead of mapping their perception on a categorical or numerical scale. Among these there may be tradeoffs in the amount of information the preference label contains and the bias associated with obtaining the label. For
example, while a graded relevance judgment on a five-point
scale may contain more information than a binary judgment,
raters may also make more errors due to the complexity
of assigning finer-grained judgments. For this reason, the paired comparison method is currently gaining growing attention, which promises assessments
that are easier and faster to obtain, less demanding task for raters, and yields more reliable data with less personal scale bias in practice. A shortcoming of paired comparison is that it has more expensive sampling complexity than the MOS test. Therefore, how to make paired comparison method efficient and applicable in reality becomes a hot topic in recent years.

%

\subsection{Crowdsourcing}

Crowdsourcing can be considered as a further development of the outsourcing principle, where tasks are submitted to an undefined and large group of people or community (a ``crowd") in the form of an open call, instead of a designated employee or subcontractor~\cite{Crowdsourcing}.  Most employers submitting tasks to an anonymous
crowd use mediators which maintain the crowd and
manage the employers¡¯ campaigns. These mediators are called crowdsourcing platforms. Among various crowdsourcing platforms, Amazon Mechanical Turk (\href{https://www.mturk.com}{MTurk}) is probably the most popular one, which provides a marketplace for a variety of tasks, and anyone who wishes to seek help from the Internet crowd can post their task requests on the website. Besides, \href{http://www.innocentive.com/}{InnoCentive}, \href{http://crowdflower.com/}{CrowdFlower},~\href{http://www.crowdrank.net/}{CrowdRank}, and \href{http://www.allourideas.org/}{AllOurIdeas} also bring the crowdsourcing revolution to various application fields.

With the help of these platforms, researchers can seek help from the Internet crowd to conduct user studies on document relevance~\cite{mm2}, document evaluation~\cite{mm25}, image annotation~\cite{mm34,MIR10}, music emotion recognition~\cite{MM13workshop1}, affection mining in computer games~\cite{MM13workshop2}, together with some studies on QoE evaluation~\cite{MM09,MM11,tmm12,MM12,conf2012-441,Keimel_etal_QoMEX2012_CrowdSourcing_Preprint}, etc. However, a major challenge of crowdsourcing QoE evaluation is that not every Internet user is trustworthy. That is, some raters try to
maximize their received payment while minimizing their own
effort and therefore submit low quality work to obtain such
a goal. Therefore, it is necessary to detect unreliable inputs and filter them out since they may cause inaccuracy in the estimation of QoE scores. For example, with complete and balanced data, TSR is proposed in~\cite{MM09} to measure the reliability of the participants' judgments. In contrast, the outlier detection method proposed in this paper is a general and simple one which could deal with not only complete and balanced paired comparison data, but also incomplete and imbalanced data.

%
%

\subsection{Statistical Ranking}

Statistical preference aggregation, in particular ranking or rating from pairwise comparisons,
is a classical problem which can be traced back to the $18^{th}$ century.
This subject area has been widely studied in various fields including the social choice or voting theory in Economics~\cite{Condorcet,Arrow51}, Psychology~\cite{Thurstone27,Saaty77}, Statistics~\cite{Noether60, David88}, Computer Vision~\cite{Yu09,Yu12,Osher11_retinex}, Information Retrieval~\cite{Pagerank,Hits}, Machine Learning~\cite{CorMohRas07,ICML14}, and others~\cite{Stefani77,Elo++,osting2013statistical}.

In particular, learning to rank trains a statistical model for ranking tasks. Popular approaches for learning to rank with pairwise comparisons include Active Ranking~\cite{ailon2012active, jamieson2011active}, IRSVM~\cite{cao2006}, RankNet~\cite{Burges2005}, and LambdaRank~\cite{Burges2006}. However, since learning to rank requires a feature vector representation
of the items to be ranked, they can not be directly applied to crowdsourced QoE evaluation.

In crowdsourced QoE evaluation, the purpose is not to predict the ranking based on features, but to aggregate a global ranking from the crowdsourcing pairwise preferences. Various methods have been proposed for crowdsourceable pairwise comparison ranking. In~\cite{chen2013pairwise}, it proposes a Bayesian framework to actively select pairwise comparison queries, and effectively combine the pairwise comparisons acquired by crowdsourcing to form a single ranking list. In~\cite{Yi2013}, it infers the preferences from crowdsourced pairwise comparison with matrix completion and compares it to collaborative filtering. In~\cite{Negahban2012}, it develops an iterative ranking aggregation algorithm from pairwise comparisons using Bradley-Terry model. Besides, there are two famous frameworks for QoE evaluation in Crowdsourcing: the \emph{Qudrant of Euphoria} by~\cite{chen2010quadrant} and  \emph{QualityCrowd} by~\cite {keimel2012challenges}.

\subsection {HodgeRank and Random Graphs}

HodgeRank, as an application of combinatorial Hodge theory to the preference or rank aggregation problem from pairwise comparison data, was first introduced in~\cite{Hodge}, inspiring a series of studies in statistical ranking~\cite{Hirani11,hodge_l1,osting2013enhanced} and game theory~\cite{Parrilo11_gameflow}, in addition to traditional applications in fluid mechanics~\cite{Chorin93} and computer vision~\cite{Yuan09_hodge,Osher11_retinex}, {etc}.

It is a general framework to decompose paired comparison data on graphs, possibly imbalanced (where different video pairs may receive different number of comparisons) and incomplete (where every participant may only give partial comparisons), into three orthogonal components. In these components HodgeRank not only provides us a mean to determine a global ranking from paired comparison data under various statistical models (e.g., Uniform, Thurstone-Mosteller, Bradley-Terry, and Angular Transform), but also measures the inconsistency of the global ranking obtained. The inconsistency shows the validity of the ranking obtained and can be further studied in terms of its geometric scale, namely whether the inconsistency in the ranking data arises locally or globally. Local inconsistency can be fully characterized by triangular cycles, while global inconsistency involves cycles consisting nodes more than three, which may arise due to data incompleteness and once presented with a large component indicates some serious conflicts in ranking data. However through random graphs, we can efficiently control global inconsistency.

Random graph is a graph generated by some random process. It starts with a set of \emph{n} vertices and adds edges between them at random. Different random graph models produce different probability distributions on graphs. Among various random graphs (i.e., the Erd\"{o}s-R\'{e}nyi random graph~\cite{ErdRen59}, random regular graph~\cite{k-regular-paper}, preferential attachment random graph~\cite{preattachment}, small world random graph~\cite {smallworld}, and geometric random graph~\cite{geo}), the most commonly studied one is the Erd\"{o}s-R\'{e}nyi random graph~\cite{ErdRen59}. It can be viewed as a random sampling process of pairs or edges independently and identically distributed (I.I.D.), and thus is well suited to crowdsourcing scenario where raters enter the test system in a dynamic and random way. In~\cite{MM11,tmm12}, a random design principle based on the Erd\"{o}s-R\'{e}nyi random graph theory is investigated to conduct crowdsourcing tests. It shows that for a large Erd\"{o}s-R\'{e}nyi random graph $G(n,p)$ with $n$ nodes and every edge sampled with probability $p$, $p\gg n^{-1}\log n$ is necessary to ensure the graph is connected and the inference of a global ranking is thus possible. To avoid global inconsistency from Hodge decomposition, it suffices to have larger sampling rates at $p\gg n^{-1/2}$. In this paper, we also focus on this simple yet powerful random graph model particularly in the scenarios where outliers are present.

\subsection{Outlier Detection}

Outliers are typically defined to be data samples that have unusual deviation from the remaining data. Hawkins formally defined in~\cite{hawkins1980identification} the concept of an outlier as follows: ``An outlier is an observation which deviates so much from the other observations as to arouse suspicions that it was generated by a different mechanism." Outliers are rare events, but once they have occurred, they may lead to a large instability of models estimated from the data. Statistical approaches were the earliest algorithms used for outlier detection, such as distribution-based, depth-based, distance-based, density-based, and clustering method~\cite{papadimitriou2003loci}. More recently, this problem has been studied quite extensively by the computer science community. In subjective quality evaluation in multimedia, there are several reasons why some user ratings are not reliable and need to be filtered out
in order to avoid false QoE results~\cite{Hossfeld2014}: the test subjects may not understand the test and the test instructions properly; wrong test conditions may occur due to errors
in the web-based test application or due to incompatibilities of the test application
with the subject's hard- and software; or the subjects do the test in a hurry resulting into sloppy work and unreliable results. Numerous efforts have been made in order to detect outliers and improve the quality of the results. In~\cite{MM13}, it formulates the outlier detection as a LASSO problem based on sparse approximations of cyclic ranking projection of paired comparison data. Then regularization paths of the LASSO problem could provide us an order on samples tending to be outliers. Such an approach is inspired by Huber's celebrated work on robust regression~\cite{Huber81}. On the other hand, recently~\cite{Yan13} proposed a fast algorithm called \emph{adaptive outlier pursuit} (AOP) for random-valued impulse noise removal, which has been applied to many applications in image and signal processing such as robust 1-bit compressive sensing~\cite{YanYO12}, robust binary fused compressive sensing~\cite{ZengF14}, and robust low rank matrix completion~\cite{YanYO13}. Such a work is based on iterative least trimmed squares. In this paper, we develop applications of AOP in the scenario of robust QoE evaluation. 

\section{Iterative Least Trimmed Squares} \label{sec:Outlier_Detection}
In this section, we propose a method for automatic outlier detection without any priori information about the number of outliers. It adaptively detects outliers and obtains robust QoE evaluation with the outlier removal. Brief introductions on robust ranking are provided before the algorithm is described.

\subsection{The Problem of Robust Ranking}


Assume that there are $m$ participants and $n$ items to be ranked. Let $Y_{ij}^\alpha$ denote the degree that participant $\alpha$ prefers item $i$ to item $j$.
Without loss of generality, one assumes that $Y_{ij}^\alpha>0$ if $\alpha$ prefers $i$ to $j$ and $Y_{ij}^{\alpha}< 0$ otherwise.
In addition, we assume that the paired comparison data is \emph{skew-symmetric} for each $\alpha$, i.e., $Y^\alpha_{ij}=-Y^\alpha_{ji}$. The strategy used in QoE evaluation can be dichotomous
choice or a $k$-point Likert scale with $k\geq3$. In this paper, we shall focus on the dichotomous choice, in which $Y_{ij}^\alpha$ can be taken as $\{\pm 1\}$ only. However, the theory can be applied to more general cases with $k$-point Likert scales. 

In subjective multimedia assessment, it is natural to assume
\begin{equation} \label{eq:linear}
Y_{ij}^\alpha = \mbox{sign}(s_i^\ast - s_j^\ast + Z_{ij}^\alpha),
\end{equation}
where $\mbox{sign}(\cdot)=\pm1$ measures the sign of the value, $\vs^*=\{s_1^*,\cdots,s_n^*\}\in \R^{n}$ is the true scaling score on $n$ items and $Z_{ij}^\alpha$ is the noise. In practice the global rating score $\vs=\{s_1,\cdots,s_n\}$ can be obtained by solving the following optimization problem
\begin{equation} \label{eq:ho_rank0}
\Min_{\vs\in {\mathbb{R}}^{n}} \sum_{i\neq j,\alpha} W_{ij}^\alpha \L(s_i - s_j, Y_{ij}^\alpha),
\end{equation}
where $\L(x,y):\R\times \R\to \R$ is a loss function depending on the distribution of the noise, $W_{ij}^\alpha$ denotes the importance weights (e.g., number of paired comparisons) on $\{i,j\}$ made by rater $\alpha$, and $s_i$ (or $s_j$) represents the global ranking score of item \emph{i} (or \emph{j}). A geometric interpretation of \eqref{eq:ho_rank0} is to look for some potential function $\vs:[n]\to \R$ whose \emph{gradient} captures main variations in paired comparison data $Y$.

If the noise is independent and identically distributed (i.i.d.), the Gauss-Markov theorem tells us that the unbiased estimator with minimal variance is obtained by the choice of square loss $\L(x,y)=(x-y)^2$. In this case the global rating score $\vs$ satisfies the normal equation:
\begin{equation}
\label{eq:HodgeRank}
\vL\vs = \vb,
\end{equation}
where $\vL=\vD-\vA$ is the unnormalized graph Laplacian defined by $A_{ij}=\sum_\alpha W_{ij}^\alpha$ and $\vD$ is the diagonal matrix with $D_{ii}=\sum_{j,\alpha} W_{ij}^\alpha$, $\vb$ is the divergence flow defined by $b_i = \sum_{j,\alpha}W_{ij}^\alpha Y_{ij}^\alpha$. Such an algorithm has been used in~\cite{MM11,tmm12,MM12} to derive scaling scores in subjective multimedia assessment. Via combinatorial Hodge decomposition \cite{Hodge,tmm12}, the residue of the least squares solution $r_{ij}^\alpha = Y_{ij}^\alpha - s_i - s_j$ can be interpreted as \emph{cyclic rankings} on $n$ items.

However, not all comparisons are trustworthy and there may be sparse outliers due to different test conditions, human errors, or abnormal variations in context. Putting in a mathematical way, here we consider
\begin{equation}
Z_{ij}^\alpha=E_{ij}^\alpha + N_{ij}^\alpha,
\end{equation}
where outlier $E_{ij}^\alpha$ has a much larger magnitude than $N_{ij}^\alpha$ and is sparse as zero with probability $p\in (0,1]$. When sparse outliers exist, (\ref{eq:ho_rank0}) becomes unstable and may give bad estimation. If the outliers can be detected and removed, then the solution from least squares on the remaining comparisons is more accurate and gives a better estimation.

In ~\cite{MM13}, the famous Huber's loss~\cite{Huber81} is chosen for robust ranking as $\L(s_i-x_j-Y_{ij}^\alpha) = \rho_\lambda(s_i - s_j - Y_{ij}^\alpha)$ where
    \begin{equation*}
\rho_\lambda(x) =
    \left\{
    \displaystyle \begin{array}{ll}
        x^2/2, & \textrm{if $|x|\leq \lambda$}\\
        \lambda |x| - \lambda^2/2, & \textrm{if $|x|> \lambda$.}
    \end{array}
    \right.
    \end{equation*}
When $|s_i - s_j - Y_{ij}^\alpha| < \lambda$, the comparison is regarded as a ``good" one with Gaussian noise and L2-norm penalty is used on the residual. Otherwise, it is regarded as a ``bad'' one contaminated by outliers and one uses L1-norm penalty which is less sensitive to the amount of deviation. Assume that the importance weights are the same ($W_{ij}^\alpha=1$). In this case, \eqref{eq:ho_rank0} is equivalent to the following LASSO problem, often called Huber-LASSO,
\begin{equation} \label{eq:hlasso}
\Min\limits_{\vs\in {\mathbb{R}}^{n},\vE}  \sum\limits_{i,j,\alpha} \frac{1}{2} (s_i - s_j- Y_{ij}^\alpha+E_{ij}^\alpha)^2+\lambda \|\vE\|_1.
\end{equation}
A simple geometric interpretation from Hodge decomposition \cite{MM13} is that the outlier $\vE$ is a sparse approximation of \emph{cyclic} ranking projection which summarizes the conflicts of interests among voters.

There are a couple of issues in such a Huber-LASSO approach~\cite{MM13}: 1) the LASSO estimator is well-known to be biased; 2) the computational cost of Huber-LASSO path is expensive as every sample is associated with an outlier indicator variable $E^\alpha_{ij}$. To solve (1), one typically exploits Huber-LASSO in outlier detection, followed by a subset least squares with only non-outlier samples. This is often called Least Trimmed Squares (LTS) in robust statistics \cite{LTS}. In the remaining of this section, we will see some iterative versions of LTS leads to fast algorithms for robust ranking which automatically finds the number of outliers in practice.

\subsection{Least Trimmed Squares}

Given $K$ as the number of outliers, the least trimmed squares model can be written as
\begin{equation} \label{eq:ho_rank_aop}
\left\{\begin{array}{rl}
\Min\limits_{\vs\in {\mathbb{R}}^{n},\Lambda}& \sum\limits_{i,j,\alpha} \Lambda_{ij}^\alpha (s_i - s_j -Y_{ij}^\alpha)^2,\\
\textnormal{subject to }&\sum\limits_{i,j,\alpha}(1-\Lambda_{ij}^\alpha)\leq K, \Lambda_{ij}^\alpha\in\{0,1\},
\end{array}\right.
\end{equation}
where $\Lambda_{ij}^\alpha$ is used to denote the outlier as follows:
\begin{align}
\Lambda_{ij}^\alpha =\left\{\begin{array}{ll}0, & \textnormal{ if }Y_{ij}^\alpha \textnormal{ is a outlier},\\
1,& \textnormal{ otherwise}.
\end{array}\right.
\end{align}
\begin{remark}
When the importance weights are not the same, we can modify the problem into
\begin{equation*}
\left\{\begin{array}{rl}
\Min\limits_{\vs\in {\mathbb{R}}^{n},\Lambda}& \sum\limits_{i,j,\alpha} \Lambda_{ij}^\alpha W_{ij}^\alpha (s_i - s_j-Y_{ij}^\alpha)^2,\\
\textnormal{subject to }&\sum\limits_{i,j,\alpha}(1-\Lambda_{ij}^\alpha)W_{ij}^\alpha\leq K, \Lambda_{ij}^\alpha\in\{0,1\}.
\end{array}\right.
\end{equation*}
\end{remark}

Let
\begin{align*}
F(\vs,\Lambda)=&\sum\limits_{i,j,\alpha} \Lambda_{ij}^\alpha (s_i - s_j- Y_{ij}^\alpha)^2\\
& +\iota_{\{\Lambda:\sum_{i,j,\alpha}(1-\Lambda_{ij}^\alpha)\leq K, \Lambda_{ij}^\alpha\in\{0,1\}\}}
\end{align*}
where $\iota_{\{\Lambda:\sum_{i,j,\alpha}(1-\Lambda_{ij}^\alpha)\leq K, \Lambda_{ij}^\alpha\in\{0,1\}\}}$ is the indicator function which equals to zero when both $\sum_{i,j,\alpha}(1-\Lambda_{ij}^\alpha)\leq K$ and $\Lambda_{ij}^\alpha\in\{0,1\}$ are satisfied and $+\infty$ otherwise, then problem~\eqref{eq:ho_rank_aop} is equivalent to
\begin{align}
\Min_{\vs\in \mathbb{R}^{n},\Lambda}F(\vs,\Lambda).
\end{align}

This is a nonconvex optimization problem. However one can split the minimization over $\Lambda$ and $\vs$ into two steps. For solving the problem in $\vs$ with $\Lambda$ fixed, it is a convex least squares problem, and the problem of finding $\Lambda$ with $\vs$ fixed can be solved in one step. These two subproblems are:

1) Fix $\Lambda$ and update $\vs$. We need to solve a least squares problem with the comparisons that are detected to be outliers removed.

2) Fix $\vs$ and update $\Lambda$. This time we are solving
\begin{equation}
\left\{\begin{array}{rl}
\Min\limits_{\Lambda}& \sum\limits_{i,j,\alpha} \Lambda_{ij}^\alpha (s_i - s_j-Y_{ij}^\alpha)^2,\\
\textnormal{subject to }&\sum\limits_{i,j,\alpha}(1-\Lambda_{ij}^\alpha)\leq K, \Lambda_{ij}^\alpha\in\{0,1\}.
\end{array}\right.
\end{equation}
This problem is to choose $K$ elements with largest summation from the set $\{(s_i - s_j- Y_{ij}^\alpha)^2\}$. Denoting $\tau$ as the value of the $K$th largest term in that set, $\Lambda$ can then be calculated by

\begin{align}\label{lambdaupdate}
\Lambda_{ij}^\alpha =\left\{\begin{array}{ll}1, & \textnormal{ if }(s_i - s_j - Y_{ij}^\alpha)^2<\tau,\\
0,& \textnormal{ otherwise}.
\end{array}\right.
\end{align}
If the $K$th and $(K+1)$th largest terms have the same value, then we can choose any $\Lambda$ such that
$\sum\limits_{i,j,\alpha}(1-\Lambda_{ij}^\alpha)\leq K$ and
\begin{align}\label{lambdaupdate2}
\min_{i,j,\alpha,\Lambda_{ij}^\alpha=0} (s_i - s_j - Y_{ij}^\alpha)^2\geq \max_{i,j,\alpha,\Lambda_{ij}^\alpha=1} (s_i - s_j - Y_{ij}^\alpha)^2.
\end{align}

Such a procedure is described precisely in the following algorithm.
\begin{algorithm}[H]
\caption{Iterative Least Trimmed Squares with $K$}	\label{alg:ilts}
	\begin{algorithmic}
		\State \textbf{Input:} $\{Y_{ij}^\alpha\}$, $K\geq0$.
		\State \textbf{Initialization:} $k=0$, $\Lambda_{ij}^\alpha=1$.
			\For{$k=1,2,\cdots$}
				\State Update $\vs^k$ by solving the least squares problem (\ref{eq:ho_rank0}) using only the comparisons with $\Lambda_{ij}^\alpha=1$.
                \State Update $\Lambda^k$ from \eqref{lambdaupdate} or \eqref{lambdaupdate2} with one different from previous ones.
			\EndFor
			\State \Return $\vs$.
	\end{algorithmic}
\end{algorithm}

%
%
%
%

Let $E(\vs)=\min_{\Lambda}F(\vs,\Lambda)$ and we have the following theorem about the convergence of Algorithm~\ref{alg:ilts}.

\begin{theorem} Algorithm~\ref{alg:ilts} will converge in finite steps and the output $\vs$ is a local minimum point of $E(\vs)$.
\end{theorem}
\begin{proof} From the algorithm, we have
\begin{align}F(\vs^k,\Lambda^k)\geq F(\vs^{k+1},\Lambda^k)\geq F(\vs^{k+1},\Lambda^{k+1}).
\end{align}
Additionally there are only finite number of $\Lambda$'s. Therefore the algorithm will stop in finite steps if the $\vs$-subproblem is solved exactly. Assume that we have $F(\vs^k,\Lambda^k)= F(\vs^{k+1},\Lambda^{k+1})$. Thus
\begin{align*}F(\vs^k,\Lambda^k)&=F(\vs^{k+1},\Lambda^k)=\min_\vs F(\vs,\Lambda^k),\\
F(\vs^k,\Lambda^k)&=\min_\Lambda F(\vs^k,\Lambda)=E(\vs^k),\end{align*}
which means that $(\vs^k,\Lambda^k)$ is a coordinatewise minimum point of $F(\vs,\Lambda)$. We will show that $\vs^k$ is a local minimum point of $E(\vs)$.

Let $\tau^k$ be the $K$th largest term of $\{(s^k_i-s^k_j-Y_{ij}^\alpha)^2\}$, and define $\Lambda_+=\{(i,j,\alpha):(s^k_i-s^k_j-Y_{ij}^\alpha)^2>\tau^k\}$ and $\Lambda_-=\{(i,j,\alpha):(s^k_i-s^k_j-Y_{ij}^\alpha)^2<\tau^k\}$. Then we can find $\epsilon>0$ such that when $\|\vs-\vs^k\|_2<\epsilon$, we have $(s_i-s_j-Y_{ij}^\alpha)^2>\tau$ for all $(i,j,\alpha)\in \Lambda_+$ and $(s_i-s_j-Y_{ij}^\alpha)^2<\tau$ for all $(i,j,\alpha)\in\Lambda_-$, where $\tau$ is the $K$th largest term of $\{(s_i-s_j-Y_{ij}^\alpha)^2\}$. Notice that $E(\vs)=\min_{\Lambda}F(\vs,\Lambda)$, then there is $\bar{\Lambda}$ such that $E(\vs)=F(\vs,\bar{\Lambda})$, with $\bar{\Lambda}_{ij}^\alpha>0$ when $(i,j,\alpha)\in \Lambda_+$ and $\bar{\Lambda}_{ij}^\alpha<0$ when $(i,j,\alpha)\in\Lambda_-$. Thus $E(\vs^k)=F(\vs^k,\bar{\Lambda})$. In addition, we have $F(\vs^k,\bar{\Lambda})\leq F(\vs,\bar{\Lambda})$ because all $\Lambda$'s satisfying~\eqref{lambdaupdate2} for $\vs^k$ are chosen before the algorithm stops. Hence, $E(\vs^k)=F(\vs^k,\bar{\Lambda})\leq F(\vs,\bar{\Lambda})=E(\vs)$, and $\vs^k$ is a local minimizer of $E(\vs)$.
\end{proof}

\subsection{Adaptive Least Trimmed Squares}
If the number of outlier $K$ is given, Algorithm~\ref{alg:ilts} can be used to detect the outliers and improve the performance of least squares. However, in practice, the exact number of outliers $K$ may be unknown. If $K$ is underestimated, some remaining outliers will still damage the performance. On the other hand, if $K$ is overestimated, too many outliers are removed, and the resulting data is not enough or too biased for QoE evaluation. For a few applications such as impulse noise removal, the number of outliers can be estimated accurately, while it is difficult for many applications including crowdsourceable QoE evaluation. Therefore, a outlier detection method which can automatically estimate the number of outliers is strongly needed.

In the following, we propose a method to estimate the number of outliers automatically. At first, when the number of outliers is unknown, we can use least squares to find an estimate of $\vs$, then the number of outliers according to this $\vs$ can be calculated, i.e., the total number of comparisons with wrong directions ($Y_{ij}^\alpha$ has different sign with $s_i-s_j$) denoted as $\widetilde{K}$. Because this $\vs$ is not accurate and $\widetilde{K}$ is an overestimation of $K$. We can underestimate the number of outliers as $\utilde{K}=\beta_1\widetilde{K}$ ({$\beta_1\in (0,1)$),  With this underestimate $\utilde{K}$, we can solve the least squares problem after the $\utilde{K}$ comparison that are considered to be outliers removed and obtain an improved $\vs$. Then we have to increase the estimation of the number $\utilde{K}$ by $\beta_2$ ($\beta_2\in (1,\infty)$), but the number can not be larger than $\widetilde{K}$, the total number of comparisons mismatching the current score, because there are only $\widetilde{K}$ outliers with the current score. Therefore the update of $\utilde{K}$ is just $\utilde{K}=\min(\lfloor\beta_2\utilde{K}\rfloor,\widetilde{K})$ where $\lfloor x \rfloor$ ($\lceil x \rceil$) is the greatest (smallest) integer no larger (less) than $x\in \R^+$. The weight $\Lambda_{ij}^\alpha$ for the new least trimmed squares problem is binary (0 or 1) and determined by $\utilde{K}$ largest outliers. Iterations go on until a fixed point is met where $\utilde{K}=\widetilde{K}$ gives the estimated number of outliers. Algorithm \ref{alg:IRLS} describes such a procedure precisely, which is called here \emph{Iterative Least Trimmed Squares without $K$} or simply \emph{Adaptive Least Trimmed Squares}.

\begin{algorithm}[h]
\caption{Adaptive Least Trimmed Squares}	
\label{alg:IRLS}
	\begin{algorithmic}
		\State \textbf{Input:} $\{Y_{ij}^\alpha\}$, $\textrm{Miter}>0$, $\beta_1<1$, $\beta_2>1$.
		\State \textbf{Initialization:} $k=0$, $\Lambda_{i,j}^\alpha(k)=1$, $\utilde{K}_k=0$.
		\For{$k=1,\cdots,$ Miter}				
		           	\State Update $\vs^k$ with least squares (\ref{eq:ho_rank0}) using only the comparisons with $\Lambda_{ij}^\alpha(k-1)=1$.
			        \State Let $\widetilde{K}_k$ be the total number of comparisons with wrong directions, i.e., $Y_{ij}^\alpha$ has different sign with $s^k_i-s^k_j$.
			        \begin{equation}
\utilde{K}_k=\left\{
\begin{array}{lr}
\lfloor\beta_1 \widetilde{K}_{k}\rfloor,& \mbox{if $k=1$}; \\
\min(\lfloor \beta_2 \utilde{K}_{k-1}\rfloor, \widetilde{K}_k),& \mbox{otherwise},
\end{array}
\right.
\end{equation}
			        \State {\bf If} $\utilde{K}_k=\widetilde{K}_k$, break.
		          	\State Update $\Lambda(k)$ using \eqref{lambdaupdate} or \eqref{lambdaupdate2} with $K=\utilde{K}_k$.
		\EndFor			
					        	\State Find $\hat{\vs}$ with least squares (\ref{eq:ho_rank0}) using only the samples with $\Lambda_{ij}^\alpha(k)=1$.
			\State \Return $\hat{\vs}$, $\hat{K}=\widetilde{K}_k$.
	\end{algorithmic}
\end{algorithm}

\begin{remark} \label{remark2}There are only two parameters to choose and they are easy to set. They are chosen according to following inequalities $\beta_1<1 < \beta_2$ ($\beta_1=0.75$ and $\beta_2=1.03$ are fixed in our numerical experiments). $\beta_1$ has to be small to make sure that the first estimation is underestimated. Then the underestimate $\utilde{K}$ is increasing geometrically with rate $\beta_2$ and $\beta_2$ can not be too large, because we do not want to increase the underestimate too much. \end{remark}

\begin{remark} \label{remark3}The algorithm is able to detect most of the outliers in our experiments with a maximum iteration number ${\rm Miter}=30$. However, there may be mistakes in the detection, and these mistakes happen mostly between two successive items in the order. Therefore, we can add one step to just compare every pair of two successive items and make the correction on the detection, i.e., if item $i$ is ranking above $j$ but the number of people choosing item $i$ over $j$ is less than the number of people choosing $j$ over $i$, we can remove those choosing $j$ over $i$ and keep those choosing $i$ over $j$.
\end{remark}

The algorithm always stops in finite steps even without a bound on ``${\rm Miter}$'', due to the following Lemma.
\begin{lemma}
 If $\widetilde{K}_k \leq C$ for $k\geq k_0$, Algorithm \ref{alg:IRLS} will stop in no more than $k^*$ steps, where
\[ k^* =\left\lceil \frac{\log C - \log \beta_1 \widetilde{K}_1}{ \log \beta_2} \right\rceil. \]
\end{lemma}
\begin{proof}
It follows from the fact that $\utilde{K}_k$ is a monotonic increasing sequence for $\beta_2>1$ and bounded $\widetilde{K}_k$.
\end{proof}

However such a result only ensures that the algorithm stops at an overestimate on the correct number of outliers. The following theorem presents a stability condition such that Algorithm \ref{alg:IRLS} returns the correct number of outliers.
\begin{theorem} \label{thm:alts} Assume that for $k\geq k_0$, every sample subset $\supp(\Lambda(k-1))$ gives an order-consistent least squares estimator $\vs^k$, i.e., $\vs^k$ induces the same ranking order as the true score ${\vs}^\ast$, then Algorithm \ref{alg:IRLS} returns the correct number of outliers in $\hat{K}$.
\end{theorem}
\begin{proof}
As $\vs^k$ is an order-consistent solution of \eqref{eq:ho_rank0}, by definition $\widetilde{K}_k$ gives the correct number of outliers, say $K^\ast$. It actually holds for all $k\geq k_0$, that $\widetilde{K}_k\equiv K^\ast$. From Lemma 1 the claim follows.
\end{proof}

Note that Theorem \ref{thm:alts} does not require $\supp(\Lambda(k-1))$ to correctly identify the outliers, but just stable estimator $\vs^k$ which does not change the order from $\vs^\ast$. In practice, this might not be satisfied easily; but as we shall see in the next section with experiments, Algorithm \ref{alg:IRLS} typically returns stable estimators which slightly deviate in local ranking order.

\section{EXPERIMENTS}\label{sec:experiments}
{\renewcommand\baselinestretch{1.5}\selectfont
\setlength{\belowcaptionskip}{3pt}
\begin{table*}
\caption{\label{tab:presicion_adap} \textbf{\emph{Precision}}s for simulated data via iLTS, 100 times repeat.}
\newsavebox{\tablebox}
\begin{lrbox}{\tablebox}
\scriptsize
\centering
\begin{tabular}{c||p{1.18cm}p{1.18cm}p{1.18cm}p{1.18cm}p{1.18cm}p{1.18cm}p{1.18cm}p{1.18cm}p{1.2cm}cccccccccc}
 \hline  \textbf{\emph {Precision} (sd)}   &\textbf{OP=5\%}  &\textbf{OP=10\%} &\textbf{OP=15\%} &\textbf{OP=20\%} &\textbf{OP=25\%} &\textbf{OP=30\%} &\textbf{OP=35\%} &\textbf{OP=40\%} &\textbf{OP=45\%} &\textbf{OP=50\%}\\
 \hline
 \hline  \textbf{SN=1000}   &0.997(0.022) &0.993(0.023) &0.993(0.015) &0.978(0.025) &0.964(0.034) &0.942(0.037) &0.893(0.053) &0.825(0.064) &0.670(0.078)  &\red{0.505(0.097)} \\
 \hline  \textbf{SN=2000}   &1.000(0) &1.000(0) &0.998(0.009) &0.999(0.005) &0.995(0.010) &0.976(0.023) &0.947(0.034)  &0.882(0.051) &0.751(0.067) &\red{0.503(0.089)}   \\
  \hline  \textbf{SN=3000}  &1.000(0) &1.000(0) &1.000(0) &0.999(0.002) &0.998(0.005) &0.991(0.013) &0.970(0.024)  &0.926(0.036) &0.811(0.060) &\red{0.502(0.090)}  \\
   \hline  \textbf{SN=4000}   &1.000(0) &1.000(0) &1.000(0) &1.000(0) &0.999(0.002) &0.995(0.010) &0.988(0.015)  &0.945(0.031) &0.829(0.059) &\red{0.498(0.098)}  \\
    \hline  \textbf{SN=5000}   &1.000(0) &1.000(0) &1.000(0) &1.000(0) &1.000(0) &0.998(0.006) &0.990(0.015)  &0.959(0.027) &0.847(0.052) &\red{0.499(0.101)}  \\
 \hline
 \end {tabular}
  \end{lrbox}
\scalebox{0.95}{\usebox{\tablebox}}
\end{table*}
\par}
{\renewcommand\baselinestretch{1.5}\selectfont
\setlength{\belowcaptionskip}{3pt}
\begin{table*}
\caption{\label{tab:presicion_lasso} \textbf{\emph{Precision}}s for simulated data via LASSO, 100 times repeat.}

\begin{lrbox}{\tablebox}
\scriptsize
\centering
\begin{tabular}{c||p{1.18cm}p{1.18cm}p{1.18cm}p{1.18cm}p{1.18cm}p{1.18cm}p{1.18cm}p{1.18cm}p{1.2cm}cccccccccc}
 \hline  \textbf{\emph{Precision} (sd)}   &\textbf{OP=5\%}  &\textbf{OP=10\%} &\textbf{OP=15\%} &\textbf{OP=20\%} &\textbf{OP=25\%} &\textbf{OP=30\%} &\textbf{OP=35\%} &\textbf{OP=40\%} &\textbf{OP=45\%} &\textbf{OP=50\%}\\
 \hline
  \hline  \textbf{SN=1000}   &0.972(0.033) &0.962(0.030) &0.958(0.025) &0.931(0.028) &0.923(0.028) &0.905(0.032) &0.863(0.040) &0.805(0.058) &0.698(0.067)  &\red{0.513(0.085)} \\
 \hline  \textbf{SN=2000}   &0.996(0.011) &0.990(0.014) &0.984(0.016) &0.970(0.022) &0.960(0.017) &0.942(0.022) &0.914(0.031)  &0.860(0.044) &0.750(0.056) &\red{0.516(0.084)}   \\
  \hline  \textbf{SN=3000}  &0.999(0.005) &0.997(0.008) &0.992(0.012) &0.981(0.016) &0.970(0.016) &0.957(0.020) &0.929(0.022)  &0.887(0.031) &0.796(0.050) &\red{0.523(0.083)}  \\
   \hline  \textbf{SN=4000}   &0.999(0.001) &0.999(0.005) &0.996(0.009) &0.990(0.011) &0.980(0.015) &0.970(0.016) &0.946(0.019)  &0.909(0.027) &0.818(0.048) &\red{0.518(0.093)}  \\
    \hline  \textbf{SN=5000}   &0.999(0.002) &1.000(0) &0.998(0.006) &0.992(0.011) &0.985(0.015) &0.972(0.016) &0.955(0.019)  &0.917(0.027) &0.837(0.038) &\red{0.525(0.088)}  \\
 \hline
 \end {tabular}
 \end{lrbox}
\scalebox{0.95}{\usebox{\tablebox}}
\end{table*}
\par}
{\renewcommand\baselinestretch{1.5}\selectfont
\setlength{\belowcaptionskip}{3pt}
\begin{table*}
\caption{\label{tab:recall_adap} \textbf{\emph{Recall}}s for simulated data via iLTS, 100 times repeat.}

\begin{lrbox}{\tablebox}
\scriptsize
\centering
\begin{tabular}{c||p{1.18cm}p{1.18cm}p{1.18cm}p{1.18cm}p{1.18cm}p{1.18cm}p{1.18cm}p{1.18cm}p{1.2cm}cccccccccc}
 \hline  \textbf{\emph{Recall} (sd)}   &\textbf{OP=5\%}  &\textbf{OP=10\%} &\textbf{OP=15\%} &\textbf{OP=20\%} &\textbf{OP=25\%} &\textbf{OP=30\%} &\textbf{OP=35\%} &\textbf{OP=40\%} &\textbf{OP=45\%} &\textbf{OP=50\%}\\
 \hline
 \hline  \textbf{SN=1000}   &1.000(0) &0.994(0.015) &0.994(0.010) &0.981(0.020) &0.969(0.024) &0.943(0.036) &0.885(0.054) &0.805(0.066) &0.653(0.080) &\red{0.438(0.093)} \\
 \hline  \textbf{SN=2000}   &1.000(0) &1.000(0) &0.999(0.006) &0.999(0.005) &0.994(0.011) &0.978(0.019) &0.947(0.032)  &0.879(0.052) &0.727(0.071) &\red{0.456(0.087)}   \\
  \hline  \textbf{SN=3000}  &1.000(0) &1.000(0) &1.000(0) &0.999(0.002) &0.998(0.005) &0.991(0.012) &0.970(0.023)  &0.925(0.037) &0.797(0.062) &\red{0.464(0.089)}  \\
   \hline  \textbf{SN=4000}   &1.000(0) &1.000(0) &1.000(0) &1.000(0) &0.999(0.003) &0.996(0.007) &0.988(0.014)  &0.946(0.030) &0.821(0.060) &\red{0.466(0.098)}  \\
    \hline  \textbf{SN=5000}   &1.000(0) &1.000(0) &1.000(0) &1.000(0) &1.000(0)  &0.998(0.006) &0.991(0.013)  &0.962(0.025) &0.842(0.052) &\red{0.470(0.100)}  \\
 \hline
 \end {tabular}
  \end{lrbox}
\scalebox{0.95}{\usebox{\tablebox}}
\end{table*}
\par}

{\renewcommand\baselinestretch{1.5}\selectfont
\setlength{\belowcaptionskip}{3pt}
\begin{table*}
\caption{\label{tab:recall_lasso} \textbf{\emph{Recall}}s for simulated data via LASSO, 100 times repeat.}

\begin{lrbox}{\tablebox}
\scriptsize
\centering
\begin{tabular}{c||p{1.18cm}p{1.18cm}p{1.18cm}p{1.18cm}p{1.18cm}p{1.18cm}p{1.18cm}p{1.18cm}p{1.2cm}cccccccccc}
 \hline  \textbf{\emph{Recall} (sd)}   &\textbf{OP=5\%}  &\textbf{OP=10\%} &\textbf{OP=15\%} &\textbf{OP=20\%} &\textbf{OP=25\%} &\textbf{OP=30\%} &\textbf{OP=35\%} &\textbf{OP=40\%} &\textbf{OP=45\%} &\textbf{OP=50\%}\\
 \hline
 \hline  \textbf{SN=1000}   &0.972(0.033) &0.962(0.030) &0.958(0.025) &0.931(0.028) &0.923(0.028) &0.905(0.032) &0.863(0.040) &0.805(0.058) &0.698(0.067) &\red{0.513(0.085)} \\
 \hline  \textbf{SN=2000}   &0.996(0.011) &0.990(0.014) &0.984(0.016) &0.970(0.022) &0.960(0.017) &0.942(0.022) &0.914(0.031)  &0.860(0.044) &0.750(0.056) &\red{0.518(0.084)}   \\
  \hline  \textbf{SN=3000}  &0.999(0.005) &0.997(0.008) &0.992(0.012) &0.981(0.016) &0.970(0.016) &0.957(0.020) &0.929(0.022)  &0.887(0.031) &0.796(0.050) &\red{0.523(0.083)}  \\
    \hline  \textbf{SN=4000}   &0.999(0.001) &0.999(0.005) &0.996(0.009) &0.990(0.011) &0.980(0.015) &0.970(0.016) &0.946(0.019)  &0.909(0.027) &0.818(0.048) &\red{0.518(0.093)}  \\
    \hline  \textbf{SN=5000}   &0.999(0.002) &1.000(0) &0.998(0.006) &0.992(0.011) &0.985(0.015) &0.972(0.016) &0.955(0.019)  &0.917(0.027) &0.837(0.038) &\red{0.525(0.088)}  \\
 \hline
 \end {tabular}
  \end{lrbox}
\scalebox{0.95}{\usebox{\tablebox}}
\end{table*}
\par}
{\renewcommand\baselinestretch{1.5}\selectfont
\setlength{\belowcaptionskip}{3pt}
\begin{table*}
\caption{\label{tab:f1_adap} \textbf{\emph{F1}} scores for simulated data via iLTS, 100 times repeat.}

\begin{lrbox}{\tablebox}
\scriptsize
\centering
\begin{tabular}{c||p{1.18cm}p{1.18cm}p{1.18cm}p{1.18cm}p{1.18cm}p{1.18cm}p{1.18cm}p{1.18cm}p{1.2cm}cccccccccc}
 \hline  \textbf{\emph{F1} (sd)}   &\textbf{OP=5\%}  &\textbf{OP=10\%} &\textbf{OP=15\%} &\textbf{OP=20\%} &\textbf{OP=25\%} &\textbf{OP=30\%} &\textbf{OP=35\%} &\textbf{OP=40\%} &\textbf{OP=45\%} &\textbf{OP=50\%}\\
 \hline
 \hline  \textbf{SN=1000}   &0.998(0.012) &0.994(0.019) &0.994(0.012) &0.980(0.022) &0.966(0.028) &0.943(0.036) &0.889(0.053) &0.815(0.064) &0.675(0.079)  &\red{0.469(0.095)} \\
 \hline  \textbf{SN=2000}   &1.000(0) &1.000(0) &0.999(0.007) &0.999(0.005) &0.994(0.010) &0.977(0.021) &0.947(0.033)  &0.880(0.051) &0.739(0.069) &\red{0.478(0.088)}   \\
  \hline  \textbf{SN=3000}  &1.000(0) &1.000(0) &1.000(0) &0.999(0.002) &0.998(0.005) &0.991(0.012) &0.970(0.023)  &0.925(0.036) &0.804(0.061) &\red{0.482(0.089)}  \\
   \hline  \textbf{SN=4000}   &1.000(0) &1.000(0) &1.000(0) &1.000(0) &0.999(0.003) &0.996(0.009) &0.988(0.014)  &0.946(0.030) &0.825(0.059) &\red{0.482(0.098)}  \\
    \hline  \textbf{SN=5000}   &1.000(0) &1.000(0) &1.000(0) &1.000(0) &1.000(0) &0.998(0.006) &0.990(0.014)  &0.960(0.026) &0.845(0.052) &\red{0.484(0.101)}  \\
 \hline
 \end {tabular}
  \end{lrbox}
\scalebox{0.95}{\usebox{\tablebox}}
\end{table*}
\par}

{\renewcommand\baselinestretch{1.5}\selectfont
\setlength{\belowcaptionskip}{3pt}
\begin{table*}
\caption{\label{tab:f1_lasso} \textbf{\emph{F1}} scores for simulated data via LASSO, 100 times repeat.}

 \begin{lrbox}{\tablebox}
\scriptsize
\centering
\begin{tabular}{c||p{1.18cm}p{1.18cm}p{1.18cm}p{1.18cm}p{1.18cm}p{1.18cm}p{1.18cm}p{1.18cm}p{1.2cm}cccccccccc}
 \hline  \textbf{\emph{F1} (sd)}   &\textbf{OP=5\%}  &\textbf{OP=10\%} &\textbf{OP=15\%} &\textbf{OP=20\%} &\textbf{OP=25\%} &\textbf{OP=30\%} &\textbf{OP=35\%} &\textbf{OP=40\%} &\textbf{OP=45\%} &\textbf{OP=50\%}\\
 \hline
 \hline  \textbf{SN=1000}   &0.972(0.033) &0.962(0.030) &0.958(0.025) &0.931(0.028) &0.923(0.028) &0.905(0.032) &0.863(0.040) &0.805(0.058) &0.698(0.067)  &\red{0.513(0.085)} \\
 \hline  \textbf{SN=2000}   &0.996(0.011) &0.990(0.014) &0.984(0.016) &0.970(0.022) &0.960(0.017) &0.942(0.022) &0.914(0.031)  &0.860(0.044) &0.750(0.056) &\red{0.516(0.084)}   \\
  \hline  \textbf{SN=3000}  &0.999(0.005) &0.997(0.008) &0.992(0.012) &0.981(0.016) &0.970(0.016) &0.957(0.020) &0.929(0.022)  &0.887(0.031) &0.796(0.050) &\red{0.523(0.083)}  \\
   \hline  \textbf{SN=4000}   &0.999(0.001) &0.999(0.005) &0.996(0.009) &0.990(0.011) &0.980(0.015) &0.970(0.016) &0.946(0.019)  &0.909(0.027) &0.818(0.048) &\red{0.518(0.093)}  \\
    \hline  \textbf{SN=5000}   &0.999(0.002) &1.000(0) &0.998(0.006) &0.992(0.011) &0.985(0.015) &0.972(0.016) &0.955(0.019)  &0.917(0.027) &0.837(0.038) &\red{0.525(0.088)}  \\
 \hline
 \end {tabular}
   \end{lrbox}
 \scalebox{0.95}{\usebox{\tablebox}}
\end{table*}
\par}

{\renewcommand\baselinestretch{1.5}\selectfont
\setlength{\belowcaptionskip}{3pt}
\begin{table*}[!htb]
\caption{\label{tab:comp_adap} Computing time for 100 runs in total on simulated data via iLTS.}
\centering
\begin{tabular}{c||cccccccccc}
\hline  time (second)  &\textbf{OP=5\%}  &\textbf{OP=10\%} &\textbf{OP=15\%} &\textbf{OP=20\%} &\textbf{OP=25\%} &\textbf{OP=30\%} &\textbf{OP=35\%} &\textbf{OP=40\%} &\textbf{OP=45\%} &\textbf{OP=50\%}\\
 \hline
 \hline  \textbf{SN=1000}   &4.38	&3.92	&3.65	&3.58	&3.54	&3.61	&3.71	&3.75	&3.66	&3.73 \\
 \hline  \textbf{SN=2000}   &6.54	&5.93	&5.62	&5.32	&5.29	&5.33	&5.28	&5.19	&5.13	&5.15   \\
  \hline  \textbf{SN=3000}  &8.86	&8.14	&7.62	&7.31	&7.11	&6.98	&7.05	&7.07	&7.15	&7.02  \\
   \hline  \textbf{SN=4000}   &11.01	&10.32	&9.67	&9.37	&8.67	&8.87	&7.82	&8.51	&8.78	&8.81  \\
    \hline  \textbf{SN=5000}  &13.23	&12.36	&12.14	&11.73	&12.04	&11.59	&11.03	&10.82	&10.79	&10.49  \\
 \hline
 \end {tabular}
\end{table*}
\par}

{\renewcommand\baselinestretch{1.5}\selectfont
\setlength{\belowcaptionskip}{3pt}
\begin{table*}[!htb]
\caption{\label{tab:comp_lasso} Computing time for 100 runs in total on simulated data via LASSO.}
\centering
\begin{tabular}{c||cccccccccc}
 \hline  time (second)   &\textbf{OP=5\%}  &\textbf{OP=10\%} &\textbf{OP=15\%} &\textbf{OP=20\%} &\textbf{OP=25\%} &\textbf{OP=30\%} &\textbf{OP=35\%} &\textbf{OP=40\%} &\textbf{OP=45\%} &\textbf{OP=50\%}\\
 \hline
 \hline  \textbf{SN=1000}   &625.14  &673.75	&690.31	&625.85	&595.65	&636.35	&592.65	&595.15	&638.65	&560.71 \\
 \hline  \textbf{SN=2000}   &905.04	&973.64	&938.37	&1017.06	&791.37	&887.72	&855.61	&825.98	&818.99	&806.25   \\
  \hline  \textbf{SN=3000}  &1116.23	&1167.45	&1184.89	&1127.83	&1118.26	&1032.88	&916.89	&952.67	&822.35	&929.75 \\
   \hline  \textbf{SN=4000}  &1158.67	&1256.82	&1305.28	&1227.76	&1161.78	&1087.81	&1016.97	&1035.82	&948.75	 &1011.45  \\
    \hline  \textbf{SN=5000}  &1288.02	&1375.14	&1368.75	&1256.89	&1228.56	&1104.32	&992.46	&976.06	&1034.12	&1077.93  \\
 \hline
 \end {tabular}
\end{table*}
\par}

A key question in the outlier detection community is how to evaluate the effectiveness of outlier detection algorithms when the ground-truth outliers are not available. In this section, we show the effectiveness of the proposed method on simulated data with known ground truth outliers and on real-world datasets without ground truth outliers. The codes for the numerical experiments and the real-world datasets can be downloaded from \href{https://code.google.com/p/irls/}{https://code.google.com/p/irls/}.

\subsection{Simulated data}
The simulated data is constructed as follows. A random total order on $n$ candidates is created as the ground-truth order. Then we add paired comparison edges $(i,j)$ randomly with preference directions following the ground-truth order. We simulate the outliers by randomly choosing a portion of the comparison edges and reversing them in preference direction. A paired comparison graph with outliers, possibly incomplete and imbalanced, is constructed.

Here we choose $n=16$, which is consistent with the real-world datasets, and make the following definitions for the experimental parameters. The total number of paired comparisons occurred on this graph is {\bf SN} (Sample Number), and the number of outliers is {\bf ON} (Outlier Number). Then the outlier percentage {\bf OP} can be obtained as {\bf ON}/{\bf SN}.

Most outlier detection algorithms adopt a tuning parameter (\emph{t}) in order to select different amount of data samples as outliers~\cite{MM13} and the number of outliers detected changes as $t$ changes. If \emph{t} is picked too restrictively, then the algorithm will miss true outlier points (false negatives). On the other hand, if the algorithm declares too many data samples as outliers, then it will lead to too many false positives. This tradeoff can be measured in terms of \textbf{\emph{precision}} and \textbf{\emph{recall}}, which are commonly used for measuring the effectiveness of outlier detection methods. Specifically, the \textbf{\emph{precision}} is defined as the percentage of reported outliers, which truly turn out to be outliers; and the \textbf{\emph{recall}} is correspondingly defined as the percentage of ground-truth outliers, which have been reported as outliers.

The proposed method iLTS (Algorithm~\ref{alg:IRLS}) is compared with LASSO~\cite{MM13} for outlier detection on the simulated data. For the ease of comparison, here we should tell LASSO in advance the exact percentage of outliers exist in the dataset.

\begin{figure*}
\setlength{\abovecaptionskip}{5pt}
 \begin{center}
 \subfigure [ ]{
\includegraphics[width=0.18\textwidth]{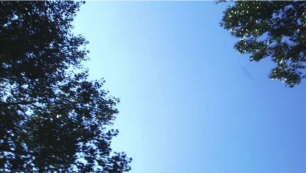}}
 \subfigure[ ]{
\includegraphics[width=0.18\textwidth]{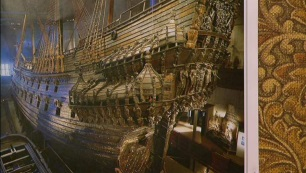}}
 \subfigure[ ]{
\includegraphics[width=0.18\textwidth]{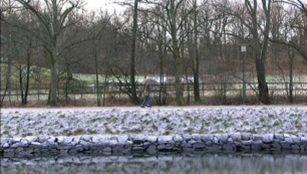}}
 \subfigure[ ]{
\includegraphics[width=0.18\textwidth]{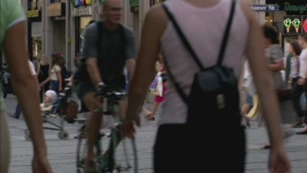}}
 \subfigure[ ]{
\includegraphics[width=0.18\textwidth]{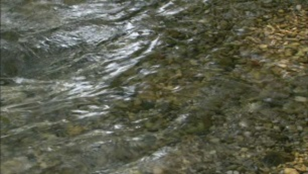}}
 \subfigure[ ]{
\includegraphics[width=0.18\textwidth]{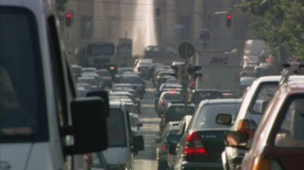}}
 \subfigure[ ]{
\includegraphics[width=0.18\textwidth]{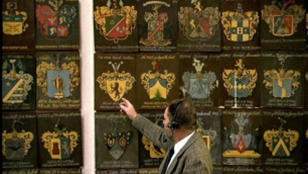}}
 \subfigure[ ]{
\includegraphics[width=0.18\textwidth]{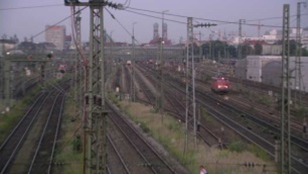}}
 \subfigure[ ]{
\includegraphics[width=0.18\textwidth]{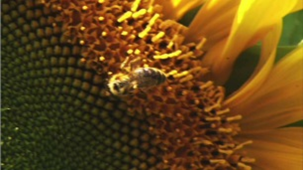}}
 \subfigure[ ]{
\includegraphics[width=0.18\textwidth]{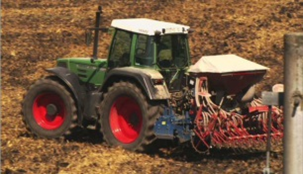}}

   \caption{Reference videos in LIVE database.} \label{datavideos}
\end{center}
\end{figure*}

The \textbf{\emph{precision}}s, \textbf{\emph{recall}}s and \textbf{\emph{F1-score}}s with standard deviations ({\bf sd}) over 100 runs for these two methods on different choices of {\bf SN} and {\bf ON} are shown in Tables~\ref{tab:presicion_adap},~\ref{tab:presicion_lasso},~\ref{tab:recall_adap},~\ref{tab:recall_lasso},~\ref{tab:f1_adap}, and~\ref{tab:f1_lasso}. \textbf{\emph{F1-score}} is a combined measure that assesses the \textbf{\emph{precision}}/\textbf{\emph{recall}} tradeoff, which reaches its best value at 1 and worst score at 0. It can be defined as follows:
\begin{align}
\textbf{\emph{F1}}  =  2\cdot\frac{\textbf{\emph{precision}}\cdot\textbf{\emph{recall}}}{\textbf{\emph{precision}}+\textbf{\emph{recall}}}.
\end{align}
When the number of outliers is not too large (i.e., {\bf OP} $\leq$ 40 \%), iLTS could produce better performance (indicated by higher \textbf{\emph{precision}}s, \textbf{\emph{recall}}s, and \textbf{\emph{F1-score}}s) than LASSO. When {\bf OP} = 50\%, i.e., half of the edges are reverted by outliers, both of these two methods show a rapid decrease of \textbf{\emph{precision}}, \textbf{\emph{recall}}, and \textbf{\emph{F1}} to about 0.5, which is the performance of random guess. It is impossible to distinguish the true signal from noise by any method when more than half of the edges are perturbed, thus a phase transition can be observed in the tables. The worse performance of iLTS for high {\bf OP}s is because the number of outliers estimated by iLTS is smaller than the exact number of outliers when the percentage of outliers is too high, which is further confirmed by the \textbf{\emph{precision}}s and \textbf{\emph{recall}}s for {\bf OP} = 50\%. When {\bf OP} = 50\%, the \textbf{\emph{recall}}s are less than 0.5 (i.e., there are more false negatives than true positives), and \textbf{\emph{precision}}s are greater than 0.5 (i.e., there are more true positives than false positives). Therefore, the number of true positives and false positives (the estimated number of outliers) is smaller than the number of true positives and false negatives (the exact number of outliers).

In addition, we compare the computing time required for these two methods to finish all the 100 runs in Tables~\ref{tab:comp_adap} and~\ref{tab:comp_lasso}. All computation is done using MATLAB R2010a on a Lenovo laptop running Windows 7 with 2.40 GHz Intel Core i7-3630QM and 8 GB 1600 MHz DDR3 memory. It is easy to see that on the simulated dataset, iLTS can achieve up to about 190 times faster than LASSO. Take (SN = 1000, OP = 15\%) as an example, LASSO needs 690.31s for 100 runs, while iLTS only needs 3.65s, which is almost 190 times faster than LASSO. Note that, in most cases iLTS could automatically determine the number of outliers exist in the dataset without \emph{a priori} knowledge.

\subsection{Real-world Data}\label{sec:VQA}

Two real-world datasets are adopted in this subsection. The first dataset PC-VQA, which is collected by~\cite{MM11}, contains 38,400 paired comparisons of the LIVE dataset~\cite{LIVE} (Figure~\ref{datavideos}) from 209 random observers. An attractive property of this dataset is that the paired comparison data is complete and balanced. As LIVE includes 10 different reference videos and 15 distorted versions of each reference video (obtained using four different distortion processes: MPEG-2 compression, H.264 compression, lossy transmission of H.264 compressed bitstreams through simulated IP networks, and lossy transmission of H.264 compressed bitstreams through simulated wireless networks), for a total of 160 videos. A round of complete comparisons for this video database requires $10 \times { 16 \choose 2 }  =1200 $ comparisons, and 38,400 comparisons correspond to 32 complete rounds.

There is no ground-truth for outliers in these real-world datasets, and we can not compute~\textbf{\emph{precision}} and~\textbf{\emph{recall}} as in the simulated data to evaluate the performance of outlier detection methods. Therefore, we inspect the outliers returned and compare them with the whole data to see if they are reasonably good outliers.

{\renewcommand\baselinestretch{1.5}\selectfont
\setlength{\belowcaptionskip}{13pt}
\begin{table*}
\caption{\label{matrixs1} Paired comparison matrices of reference (a) in PC-VQA dataset. Red numbers are outliers obtained by both iLTS and LASSO. Open blue circles are those obtained by LASSO but not iLTS, while filled blue circles are obtained by iLTS but not LASSO.}
\centering
\begin{tabular}{|c|c|c|c|c|c|c|c|c|c|c|c|c|c|c|c|c|}
\hline Video ID    &\textbf{1} &\textbf{9} &\textbf{10} &\textbf{13} &\textbf{7} &\textbf{8} &\textbf{11} &\textbf{14} &\textbf{15} &\textbf{3} &\textbf{12} &\textbf{4} &\textbf{16} &\textbf{5} &\textbf{6} &\textbf{2} \\
\hline \textbf{1}    & 0   &  22    & 29    & 30    & 30    & 29    & 29    & 29    & 30    & 28    & 29    & 32    & 32    & 31   &  32   &  31 \\
\hline \textbf{9}     &\red{10}    &  0   & 22    & 20    & 14    & 23    & 23    & 25    & 29    & 29    & 32    & 30    & 29    & 30    & 29    & 31 \\
\hline \textbf{10}    &\red{3}    &\red{10}     & 0    & 22    & 11    & 21   &  29    & 23    & 31    & 27    & 31   &  30   &  32    & 30    & 32    & 31 \\
\hline \textbf{13}    &\red{2}     &\red{12}    &\red{10}     & 0    & 18   &  22   &  23    & 27   &  31   &  28    & 29   &  29    & 29    & 25    & 27    & 28 \\
\hline \textbf{7}    &\red{2}    &\red{18}   &\red{21}   &\red{14}    &  0   &  21    & 14    & 16   &  28    & 23    & 31    & 25   &  19   &  27    & 26    & 28 \\
\hline \textbf{8}   &\red{3}     &\red{9}   &\red{11}   &\red{10}    &\red{11}     & 0    & 25   &  14    & 28   &  25    & 29   &  27    & 24   &  25    & 28    & 32  \\
\hline \textbf{11}    &\red{3}     &\red{9}      &\red{3}     &\red{9}    &\red{18}     &\red{7}     & 0    & 22    & 27    & 26    & 26    & 30   &  30   &  27   &  27    & 31  \\
\hline \textbf{14}   &\red{3}     &\red{7}       &\red{9}      &\red{5}    &\red{16}   &\red{18}    &\red{10}    &  0    & 28    & 27    & 18   &  29   &  29   &  26    & 28   &  29 \\
\hline \textbf{15}    &\red{2}     &\red{3}     &\red{1}    &\red{1}     &\red{4}      &\red{4}      &\red{5}     &\red{4}     &  0    & 25   &  20    & 22    & 26    & 25   &  29    & 24 \\
\hline \textbf{3}    &\red{4}     &\red{3}    &\red{5}    &\red{4}     &\red{9}    &\red{7}      &\red{6}    &\red{5}      &\red{7}    &  0    &\red{11}   &\textcolor{white}{\pgftextcircledblk{\bf 15}}    & 26    & 24    & 29    & 28  \\
\hline \textbf{12}   &\red{3}    &  0    &\red{1}    &\red{3}     &\red{1}     &\red{3}     &\red{6}    &\red{14}     &\red{12}   &  21     & 0    & 16    & 20    & 24    & 26    & 26 \\
\hline \textbf{4}   & 0     &\red{2}     &\red{2}     &\red{3}     &\red{7}     &\red{5}     &\red{2}     &\red{3}     &\red{10}    & \textcolor{blue}{\pgftextcircled{\bf 17}}    &\red{16}    &  0   &  15    & 26    & 27    & 30 \\
\hline \textbf{16}   & 0     &\red{3}     & 0     &\red{3}    &\red{13}     &\red{8}     &\red{2}    &\red{3}     &\red{6}     &\red{6}   &\red{12}    &\red{17}     & 0   &  22    & 24    & 28 \\
\hline \textbf{5}  &\red{1}    &\red{2}    &\red{2}     &\red{7}    &\red{5}     &\red{7}     &\red{5}   &\red{6}    &\red{7}     &\red{8}    &\red{8}    &\red{6}    &\red{10}    & 0   &  26    & 27  \\
\hline \textbf{6}   & 0     &\red{3}    &  0     &\red{5}     &\red{6}     &\red{4}     &\red{5}   &\red{4}    &\red{3}     &\red{3}     &\red{6}      &\red{5}    &\red{8}   &\red{6}    &  0   &  21  \\
\hline \textbf{2}    &\red{1}    &\red{1}    &\red{1}    &\red{4}     &\red{4}     & 0     &\red{1}     &\red{3}     &\red{8}     &\red{4}    &\red{6}      &\red{2}      &\red{4}     &\red{5}    &\red{11}    & 0 \\
\hline
\end {tabular}
\end{table*}
 \par}

{\renewcommand\baselinestretch{1.5}\selectfont
\setlength{\belowcaptionskip}{13pt}
\begin{table*}\caption{\label{matrixs2} Paired comparison matrices of reference (c) in PC-IQA dataset. Red numbers, open blue circles, and filled blue circles carry the same meanings with Table~\ref{matrixs1}.}
\centering
\begin{tabular}{|c|c|c|c|c|c|c|c|c|c|c|c|c|c|c|c|c|}
\hline \textbf{Image ID}      &\textbf{1} &\textbf{8} &\textbf{16} &\textbf{2} &\textbf{3} &\textbf{11} &\textbf{6} &\textbf{12} &\textbf{9} &\textbf{14} &\textbf{5} &\textbf{13} &\textbf{7} &\textbf{10} &\textbf{15 }&\textbf{4} \\
\hline \textbf{1}    &0	&13	&9	&16	&19	&12	&15	&13	&14	&14	&14	&17	&16	&17	&16	&16 \\
\hline \textbf{8}     &\red{6}	&0	&8	&7	&8	&5	&13	&7	&7	&8	&19	&8	&15	&9	&12	&15 \\
\hline \textbf{16}    &\red{4}	&0	&0	&9	&11	&9	&8	&15	&3	&18	&16	&17	&12	&7	&21	&18\\
\hline \textbf{2}    &\red{5}	&\red{5}	&\red{6}	&0	&8	&9	&10	&11	&7	&14	 &13	&14	&14	&13	&14	 &15 \\
\hline \textbf{3}     &\red{3}	&\red{4}	&\red{6}	&\red{7}	&0	&6	 &11	&9	&10	&16	&12	 &15	&14	&14	&18	&13 \\
\hline \textbf{11}    &\red{4}	&\red{6}	&\red{3}	&\red{5}	 &\red{6}	&0	 &\textcolor{white}{\pgftextcircledblk{\bf 5}}	&3	 &5	&6	&21	&5	 &11	 &7	 &12	 &18\\
\hline \textbf{6}   &0	&\red{2}	&\red{7}	&\red{4}	&\red{2}	&\textcolor{blue}{\pgftextcircled{\bf 7}}	&0	 &12	&12	&7	 &22	&15	 &17	&13	&13	&17   \\
\hline \textbf{12}   &\red{3}	&\red{4}	&\red{1}	&\red{4}	 &\red{4}	 &\red{3}	&\red{1}	&0	 &8	 &15	 &18	&12	&9	 &8	&13	&17 \\
\hline \textbf{9}   &\red{1}	&\red{3}	&\red{3}	&\red{5}	 &\red{1}	 &\red{3}	&\red{1}	&0	 &0	 &5	 &18	&10	&14	&9	 &7	&16  \\
\hline \textbf{14}   &0	&0	&\red{1}	&0	&0	&\red{3}	&\red{7}	 &\red{2}	 &\red{1}	&0	&14	&15	&10	&8	&17	&19 \\
\hline \textbf{5}    &0	&0	&0	&0	&0	&0	&0	&0	&0	&\red{1}	&0	&14	&19	&19	&15	&17 \\
\hline \textbf{13}   &0	&0	&0	&0	&0	&0	&0	&0	&0	&0	&\red{6}	&0	&5	&7	&17	&16  \\
\hline \textbf{7}    &0	&0	&0	&0	&0	&0	&0	&0	&0	&0	&0	&\red{5}	&0	&8	&9	&18 \\
\hline \textbf{10}    &0	&0	&0	&0	&0	&0	&0	&0	&0	&0	&0	&\red{2}	&\red{2}	&0	 &\textcolor{white}{\pgftextcircledblk{\bf 3}}	 &11  \\
\hline \textbf{15}   &0	&0	&0	&0	&0	&0	&0	&0	&0	&0	&0	&0	&0	&\textcolor{blue}{\pgftextcircled{\bf 5}}	&0	&11 \\
\hline \textbf{4}  &0	&0	&0	&0	&0	&0	&0	&0	&0	&0	&0	&\red{1}	&0	&\red{6}	 &\red{6}	&0\\\hline
\end {tabular}
\end{table*}
 \par}

We take reference (a) in the PC-VQA dataset as an illustrative example (other reference videos exhibit similar results). We compare iLTS and LASSO again in the real-world datasets. The number of outliers estimated by iLTS is used for LASSO to choose regularization parameters and select the outliers. Outliers detected by both methods are shown in the paired comparison matrix in Table~\ref{matrixs1}. The paired comparison matrix is constructed as follows (Table~\ref{matrixs2} is constructed in the same way). For each video pair $\{i,j\}$, let $n_{ij}$ be the number of comparisons, among which $a_{ij}$ raters agree that the quality of $i$ is better than $j$ ($a_{ji}$ carries the opposite meaning). So $a_{ij} + a_{ji} = n_{ij}$ if no tie occurs, and in the PC-VQA dataset, $n_{ij}\equiv32$ for all videos. The order of the video ID in this Table is arranged from high to low according to the global ranking score calculated by the least squares method (\ref{eq:ho_rank0}). The outliers picked out by both methods are mainly distributed in the lower left corner of this matrix, which implies that the outliers are those preference orders with a large deviation from the global ranking scores by L2. The total number of outliers estimated by iLTS from this reference video is 761, so the outlier percentage (\textbf{OP}) = 761/3840 = 18.65\%. For comparison, we also inspect the top 18.65\% returned by LASSO. It is easy to see that outliers returned by iLTS and LASSO are almost the same except one pair (ID = 3 and ID = 4). In the dataset, 15 raters agree that the quality of ID = 3 is better than that of ID = 4, while 17 raters have the opposite opinion. iLTS treats $a_{3,4} = 15$ as outliers, while LASSO chooses the opposite direction (i.e., treats $a_{4,3} = 17$ as outliers). LASSO tends to choose outliers as the large deviation from the gradients of global ranking scores while iLTS prefers to choose the minority in a paired comparison data. Such a small difference only leads to a local order change of nearby ranked items, ID = 3 and ID = 4. Therefore the ranking algorithms are stable.

The global ranking scores of these three algorithms, namely L2, LASSO, and iLTS, are shown in Table~\ref{tab:data5-rank}. Removing the top 18.65\% outliers in both LASSO and iLTS changes the orders of some competitive videos. Both LASSO and iLTS think ID = 12 has better performance than ID = 3 and ID = 4. The scores of LASSO and iLTS are quite similar except that the orders and scores of ID = 3 and ID = 4 are exchanged, because LASSO and iLTS choose different preference directions as outliers.

\begin{figure*}
\begin{center}
\subfigure [ ]{
\includegraphics[width=0.15\textwidth]{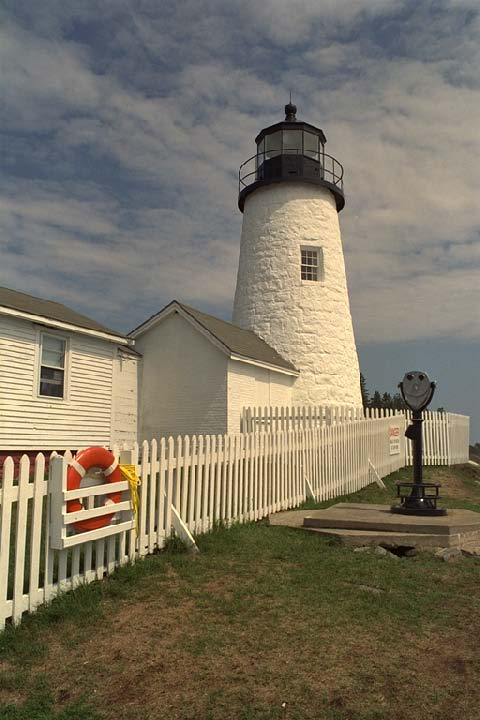}}
\subfigure [ ]{
\includegraphics[width=0.15\textwidth]{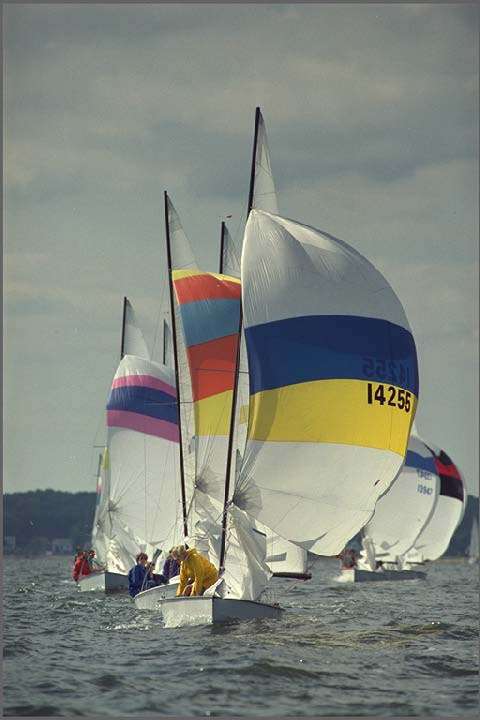}}
\subfigure [ ]{
\includegraphics[width=0.15\textwidth]{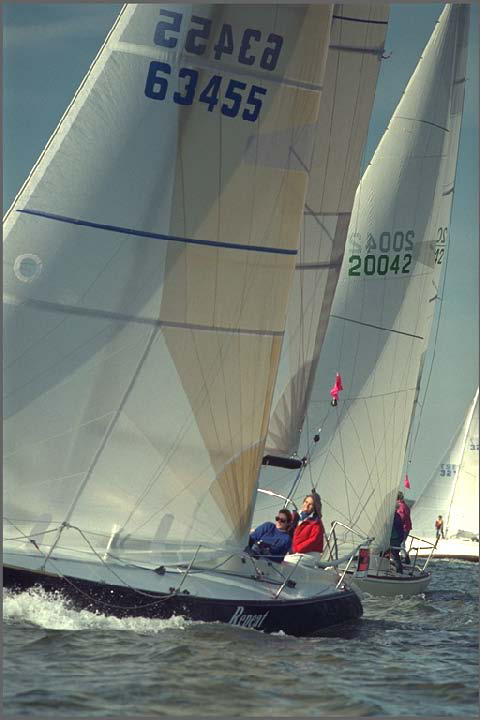}}
\subfigure [ ]{
\includegraphics[width=0.15\textwidth]{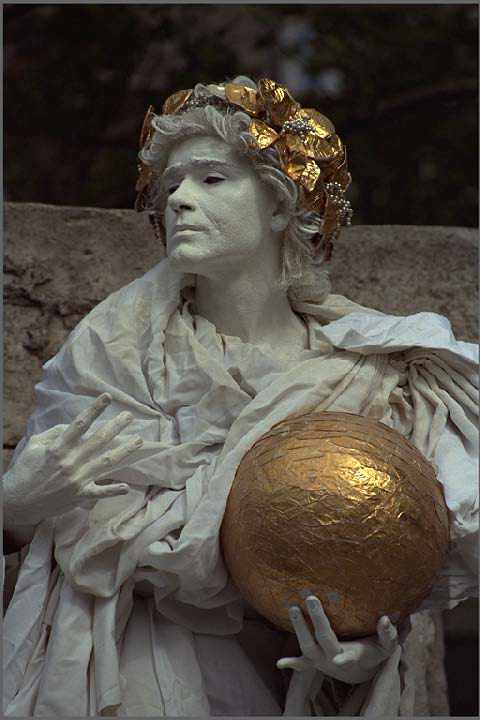}}
\subfigure [ ]{
\includegraphics[width=0.15\textwidth]{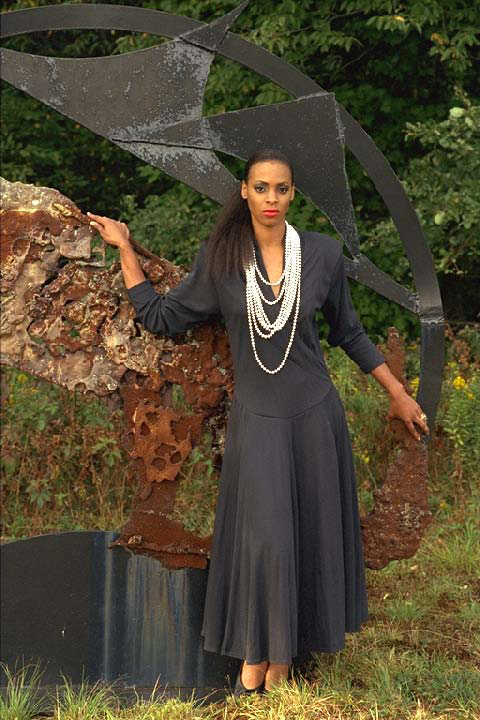}}
\subfigure [ ]{
\includegraphics[width=0.15\textwidth]{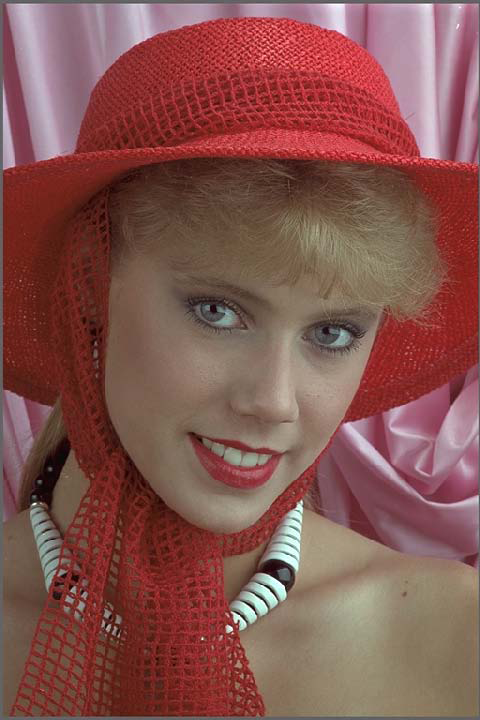}}
\subfigure [ ]{
\includegraphics[width=0.15\textwidth]{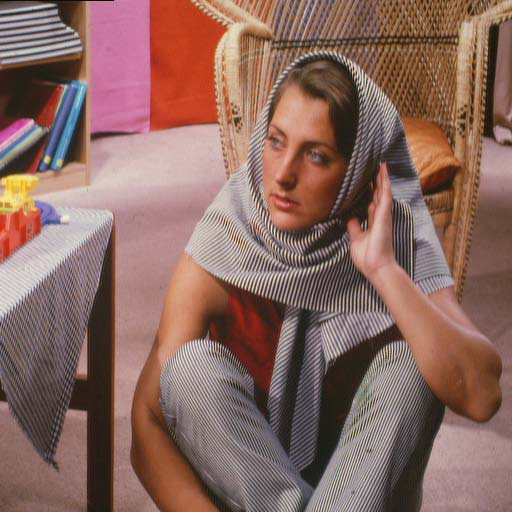}}
\subfigure [ ]{
\includegraphics[width=0.15\textwidth]{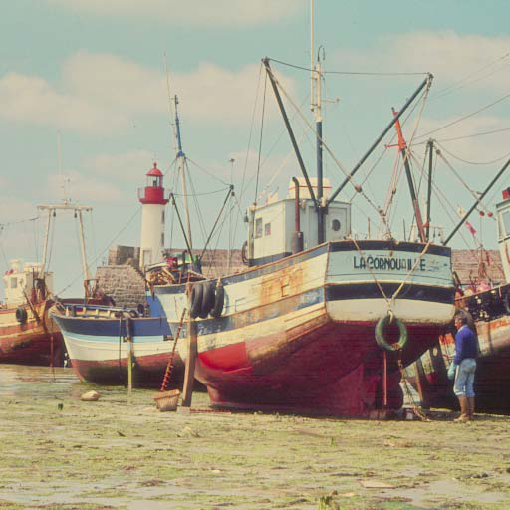}}
\subfigure [ ]{
\includegraphics[width=0.15\textwidth]{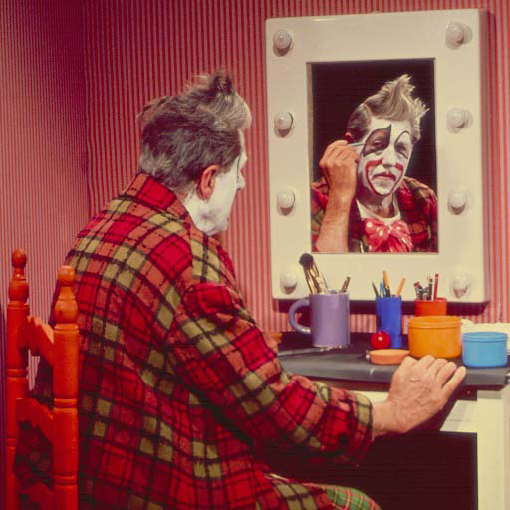}}
\subfigure [ ]{
\includegraphics[width=0.15\textwidth]{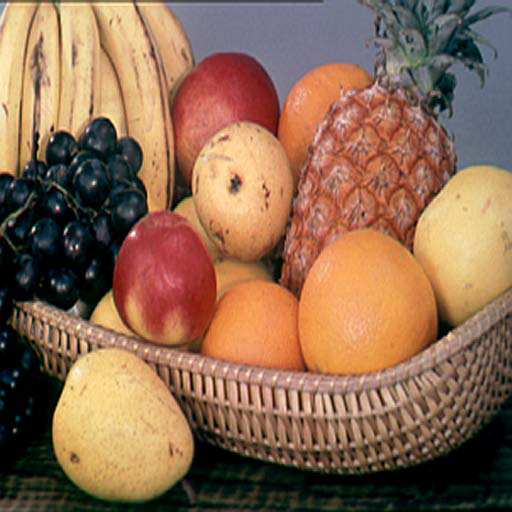}}
\subfigure [ ]{
\includegraphics[width=0.15\textwidth]{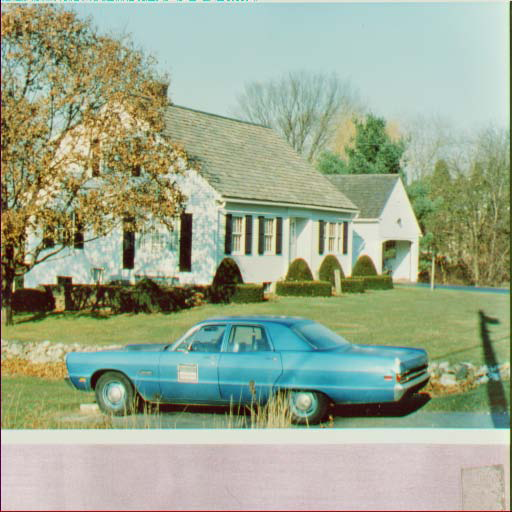}}
\subfigure [ ]{
\includegraphics[width=0.15\textwidth]{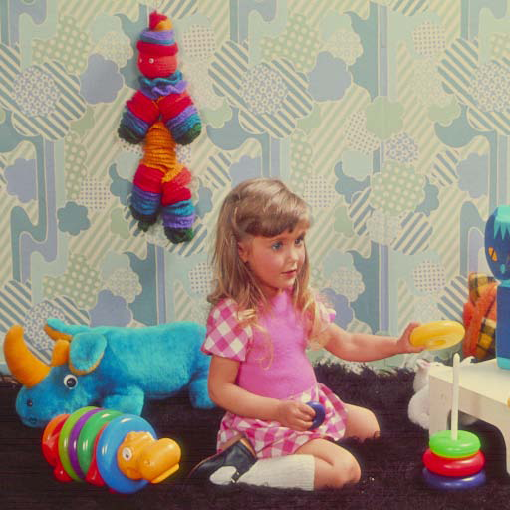}}
\subfigure [ ]{
\includegraphics[width=0.15\textwidth]{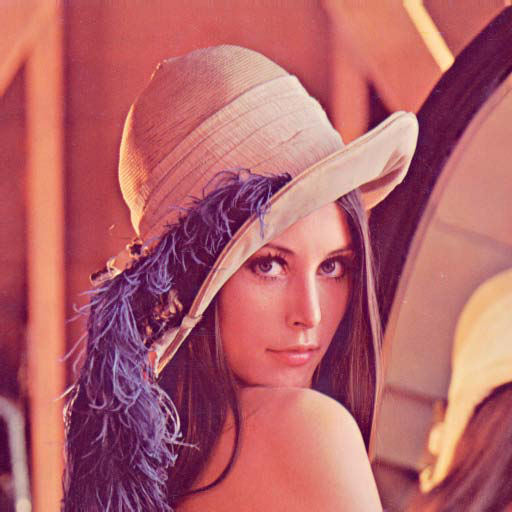}}
\subfigure [ ]{
\includegraphics[width=0.15\textwidth]{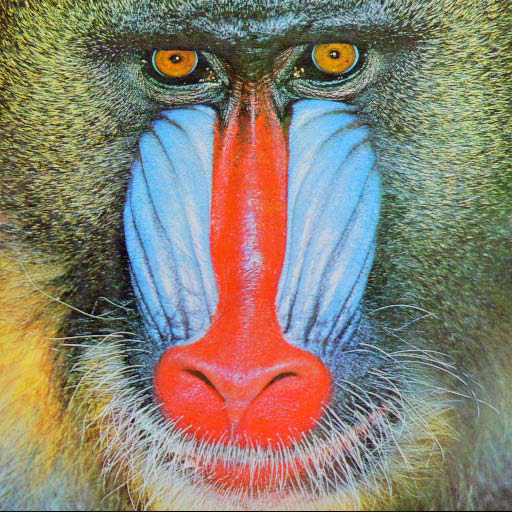}}
\subfigure [ ]{
\includegraphics[width=0.15\textwidth]{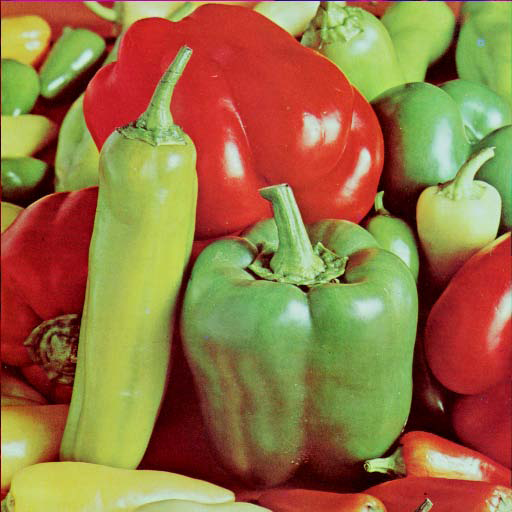}}

\caption{Reference images in the LIVE and IVC datasets. (The first six images are from the LIVE dataset and the remaining nine images are from the IVC dataset.)} \label{dataimages}
\end{center}
\vskip -0.2in
\end{figure*}

{\renewcommand\baselinestretch{1.3}\selectfont

\begin{table}[ht] \caption{\label{matrixs} Comparison of different rankings. Three ranking methods are compared with the integer representing the ranking position and the number in parentheses representing the global ranking score returned by the corresponding algorithm.}

\centering

\subtable[  Reference (a) in the PC-VQA dataset]{
      \begin{tabular}{|c|c|c|c|}
    \hline
    \textbf{Video ID} & \textbf{L2} & \textbf{LASSO} & \textbf{iLTS}\\
    \hline
 1 & 1 ( 0.7930 ) & 1 ( 0.9123 ) & 1 ( 0.9129 ) \\
9  & 2 ( 0.5312 ) & 2 ( 0.7537 ) & 2 ( 0.7539 ) \\
10 & 3 ( 0.4805 ) & 3 ( 0.6317 ) & 3 ( 0.6322 )  \\
13 & 4 ( 0.3906 ) & 4 ( 0.5522 ) & 4 ( 0.5524 )   \\
7 & 5 ( 0.2852 ) & 5 ( 0.4533 ) & 5 ( 0.4537 )   \\
8 & 6 ( 0.2383 )  & 6 ( 0.3159 )  & 6 ( 0.3163 )  \\
11 & 7 ( 0.2148 ) &  7 ( 0.2113 ) & 7 ( 0.2120 )   \\
14 & 8 (  0.1641 ) & 8  ( 0.1099 ) & 8 ( 0.1103 )  \\
15 & 9 ( -0.1758 ) & 9 ( -0.1024 ) & 9 ( -0.1029 )   \\
3 & 10 ( -0.2227 ) & \red{11 ( -0.3195 )} & \red{12 ( -0.3999 )}  \\
12 & 11 ( -0.2500 ) & \red{10 ( -0.2149 )} & \red{10 ( -0.2158 )}   \\
4 & 12 ( -0.2930 ) & \red{12 ( -0.4054 )} & \red{11 ( -0.3252 )}   \\
16 & 13 ( -0.3633 ) & 13 ( -0.5311 ) & 13 ( -0.5332 )   \\
5 & 14 ( -0.4414 ) & 14 ( -0.6573 ) & 14 ( -0.6568 )   \\
6 & 15 ( -0.6289 ) & 15 ( -0.8054 ) & 15 ( -0.8057 )   \\
2 & 16 ( -0.7227 ) & 16 ( -0.9046 ) & 16 ( -0.9042 )   \\

    \hline

\end{tabular}
       \label{tab:data5-rank}
}

\subtable[  Reference (c) in the PC-IQA dataset]{
       \begin{tabular}{|c|c|c|c|}
    \hline
    \textbf{Image ID} & \textbf{L2} & \textbf{LASSO} & \textbf{iLTS}\\
    \hline
 1 & 1 ( 0.7575 ) & 1 ( 0.9015 ) & 1 ( 0.9022 ) \\
8  & 2 ( 0.5670 ) & 2 ( 0.7088 ) & 2 ( 0.7129 ) \\
16 & 3 ( 0.5124 ) & 3 ( 0.6472 ) & 3 ( 0.6504 )  \\
2 & 4 ( 0.4642 ) & 4 ( 0.5242 ) & 4 ( 0.5248 )   \\
3 & 5 ( 0.4423 ) & 5 ( 0.4119 ) & 5 ( 0.4148 )   \\
11 & 6 ( 0.3277 ) &  6 ( 0.2592 ) &  \red{ 7 ( 0.1763 )}  \\
6 & 7 ( 0.3128 ) & 7 ( 0.2515 ) &  \red{ 6 ( 0.3124 )}   \\
12 & 8 ( 0.2423 ) & 8 ( 0.1209 ) & 8 ( 0.1261 )  \\
9 & 9 ( 0.1453 ) & 9 ( 0.0043 ) & 9 ( 0.0069 )   \\
14 & 10 ( -0.0455 ) & 10 ( -0.1274 ) & 10 ( -0.1243 )   \\
5 & 11 ( -0.3376 ) & 11 ( -0.3205 ) & 11 ( -0.3214 )   \\
13 & 12 ( -0.4785 ) & 12 ( -0.4621 ) & 12 ( -0.4560 )   \\
7 & 13 ( -0.5396 ) &  13 ( -0.5515 ) & 13 ( -0.5494 )  \\
10 & 14 ( -0.7486 ) &  14 ( -0.7005 ) & \red{ 15 ( -0.7485 )}   \\
15 & 15 ( -0.7658 ) & 15 ( -0.7511 ) & \red{ 14 ( -0.7106 )}    \\
4 & 16 ( -0.8559 ) & 16 ( -0.9163 ) & 16 ( -0.9166 )   \\

    \hline

\end{tabular}
       \label{tab:data10-rank}
}
\end{table}
\par}

The effectiveness of iLTS is demonstrated on a complete and balanced dataset, and we want to show the effectiveness of iLTS on incomplete and imbalanced datasets. The PC-IQA dataset is taken into consideration. This dataset contains 15 reference images and 15 distorted versions of each reference image, for a total of 240 images, which come from two publicly available datasets: LIVE~\cite{LIVE} and IVC~\cite{IVC} (Figure~\ref{dataimages}). The distorted images in the LIVE dataset~\cite{LIVE} are obtained using five different distortion processes: JPEG2000, JPEG, White Noise, Gaussian Blur, and Fast Fading Rayleigh, while the distorted images in the IVC dataset~\cite{IVC} are derived from four distortion types --- JPEG2000, JPEG, LAR Coding, and Blurring. Totally, 186 observers, each of whom performs a varied number of comparisons via Internet, provide 23,097 paired comparisons for subjective IQA.

Table~\ref{matrixs2} shows the comparable experimental results of iLTS vs. LASSO on a randomly selected reference image (image (c) in Figure~\ref{dataimages}). Similar observations as above can be made and we note that outliers distributed on this dataset are much sparser than PC-VQA, shown by many zeros in the lower left corner of the paired comparison matrix. Outlier percentage (OP) returned by iLTS is 173/1655 = 10.45\% in Table~\ref{matrixs2}, and it is easy to find that the detection results of LASSO vs. iLTS are different on two pairs: 1) ID = 6 and ID = 11; 2) ID = 10 and ID = 15. Similar to the last experiment, iLTS prefers to choose the minority in paired comparisons, i.e., the 5 in $11\succ 6$ and the 3 in $10\succ 15$, while LASSO selects outliers as the large deviation from the gradients of global ranking scores even when the votings are in majority. Such a difference leads to a local order change of involved items which are adjacent in ranking list, exhibiting stability in global rankings, as shown in Table~\ref{tab:data10-rank}.

\subsection{Discussion}
As we have seen in the numerical experiments, iLTS and LASSO mostly find the same outliers and when they disagree, iLTS tends to choose the minority and LASSO prefers to choose outliers as the large deviation from the gradients of global ranking
scores even when the votings are in majority. When outliers consist of minority voting as in simulated experiments, iLTS may perform better. Besides, iLTS tents to choose fewer outliers to make sure that there are no outliers in the remaining comparisons. This can also be explained from the algorithm. We choose a small initial estimation for the number of outliers, and increase this estimation until there is no outliers in the remaining comparisons. The parameter $\beta_2>1$ is chosen to be small so we will not overestimate the number of outliers too much.

%

Finally, we would like to point out that subject-based outlier detection can be a straightforward extension from our proposed iLTS. From the detection results of iLTS, one may evaluate the reliability of one participant based on all the comparisons from the participant, and drop unreliable participants.

\section{CONCLUSIONS}\label{sec:conclusions}
In this paper, we have proposed a fast and adaptive algorithm iLTS for outlier detection and robust ranking in QoE evaluation. It achieves up to 190 times faster than LASSO in outlier detection. Moreover, this method can automatically estimate the number of outliers and detect them without any priori information about the number of outliers existing in the dataset. The effectiveness and efficiency of iLTS is demonstrated on both simulated examples and real-world applications. iLTS exhibits comparable accuracy to LASSO in outlier detection. There are small distinctions between them indicating that iLTS prefers to choose minority voting data as outliers, while the LASSO selects large deviations from the gradient of global ranking score as outliers even when they are in majority voting. In both cases, the global rankings obtained are stable. A future direction is to understand under what kind of conditions such an adaptive least trimmed squares algorithm works.

In summary, we expect that the proposed iLTS for QoE evaluations will be a helpful tool for people in the multimedia community exploiting crowdsourceable paired comparison data for robust ranking.

\ifCLASSOPTIONcaptionsoff
  \newpage
\fi


\bibliographystyle{IEEEtran}
\bibliography{sigproc}
\end{document}